\documentclass[12pt,onecolumn,draftcls]{IEEEtran}

\usepackage{xifthen}
\usepackage{etoolbox}
\usepackage{mathtools}
\usepackage{graphicx}
\usepackage{subcaption}

\usepackage{amsthm,amsmath,amssymb}

\newcommand{\Var}{\mathsf{Var}}
\newcommand{\rV}[1]{{\mathrm{#1}}}
\newcommand{\rvVec}[1]{\pmb{\mathrm{#1}}}
\newcommand{\rvMat}[1]{\pmb{\mathsf{#1}}}

\newcommand{\SNR}{\gamma}
\usepackage{algorithm, algorithmicx, algpseudocode}
\newcommand{\NND}{\text{NND}}

\newcommand{\mbs}[1]{\pmb{#1}}
\newcommand{\vect}[1]{{\lowercase{\mbs{#1}}}}
\newcommand{\mat}[1]{{\uppercase{\mbs{#1}}}}

\newcommand{\T}{{\scriptscriptstyle\mathsf{T}}}
\renewcommand{\H}{{\scriptscriptstyle\mathsf{H}}}

\newcommand{\cond}{\,\vert\,}
\renewcommand{\Re}[1][]{\ifthenelse{\isempty{#1}}{\operatorname{Re}}{\operatorname{Re}\left(#1\right)}}
\renewcommand{\Im}[1][]{\ifthenelse{\isempty{#1}}{\operatorname{Im}}{\operatorname{Im}\left(#1\right)}}

\newcommand{\bv}{\vect{b}}

\newcommand{\vv}{\vect{v}}

\newcommand{\xv}{\vect{x}}
\newcommand{\yv}{\vect{y}}
\newcommand{\zv}{\vect{z}}
\newcommand{\zerov}{\vect{0}}

\newcommand{\thetav}{\vect{\theta}}

\newcommand{\Am}{\mat{a}}

\newcommand{\Cm}{\mat{c}}
\newcommand{\Dm}{\mat{d}}

\newcommand{\Hm}{\mat{h}}

\newcommand{\Qm}{\mat{q}}

\newcommand{\Um}{\mat{u}}
\newcommand{\Vm}{\mat{V}}
\newcommand{\Wm}{\mat{w}}

\newcommand{\Id}{\mat{i}}

\newcommand{\Thetat}[1][]{\ifthenelse{\isempty{#1}}{\Theta_{\text{t}}}{\Theta_{\text{t},{#1}}}}
\newcommand{\Thetar}[1][]{\ifthenelse{\isempty{#1}}{\Theta_{\text{r}}}{\Theta_{\text{r},{#1}}}}

\newcommand{\nt}{n_{\text{t}}}
\newcommand{\nr}{n_{\text{r}}}

\newcommand{\LambdaR}{\mbs{\Lambda}_{\text{R}}}
\newcommand{\LambdaT}{\mbs{\Lambda}_{\text{T}}}

\newcommand{\rvLambdaR}{\rvMat{\Lambda}_{\text{R}}}
\newcommand{\rvLambdaT}{\rvMat{\Lambda}_{\text{T}}}
\newcommand{\rvA}{\rvMat{A}}

\newcommand{\xML}{\hat{\pmb{x}}_{\mathrm{ML}}}
\newcommand{\xMMSE}{\hat{\pmb{x}}_{\mathrm{LMMSE}}}
\newcommand{\xMLnaive}{\hat{\pmb{x}}^{0}_{\mathrm{ML}}}
\newcommand{\likelihood}[3]{p(#1 \cond #2, #3)}
\newcommand{\E}[1][]{\ifthenelse{\isempty{#1}}{\mathbb{E}}{\mathbb{E}\left[#1\right]}}
\newcommand{\PP}[1][]{\ifthenelse{\isempty{#1}}{\mathbb{P}}{\mathbb{P}\left[#1\right]}}
\newcommand{\indic}{\mathbbm{1}}
\newcommand{\xMLapp}{\hat{\pmb{x}}_{\mathrm{aML}}}

\newcommand{\diag}{\mathrm{diag}}
\newcommand{\llh}{f}
\newcommand{\allh}{\hat{f}}
\newcommand{\Amx}{\Am}
\newcommand{\bvx}{\bv}
\newcommand{\Wmx}{\Wm_{\!\!\xv}}
\newcommand{\Wmxhat}{\Wm_{\!\!\hat{\xv}}}
\newcommand{\Wmxhatnaive}{\Wm_{\!\!\xMLnaive}}

\usepackage{graphicx}
\usepackage{amsmath}
\usepackage{color}
\usepackage{epsfig}
\usepackage{amssymb}

\usepackage{bbm,color}

\title{An Approximate ML Detector for MIMO Channels Corrupted by Phase Noise} 

\author{Richard Combes and Sheng Yang
\thanks{The authors are with the Laboratory of Signals and Systems~(L2S), CentraleSup\'elec, 3 rue Joliot-Curie, 91190
Gif-sur-Yvette, France.~(e-mail: \texttt{richard.combes,sheng.yang@centralesupelec.fr})}
}

\newtheorem{proposition}{Proposition}

\newtheorem{lemma}{Lemma}

\newtheorem{remark}{Remark}

\allowdisplaybreaks

\begin{document}

\maketitle

\begin{abstract}
  We consider the multiple-input multiple-output~(MIMO) communication channel impaired by phase noises
  at both the transmitter and receiver. We focus on the maximum likelihood~(ML) detection problem for uncoded
  single-carrier transmission. We derive an approximation of the likelihood function, based on which we propose an
  efficient detection algorithm. The proposed algorithm, named \emph{self-interference whitening}~(SIW), consists in
  1)~estimating the self-interference caused by the phase noise perturbation, then~2)~whitening the said interference,
  and finally~3)~detecting the transmitted vector.  While the exact ML solution is computationally intractable, we
  construct a simulation-based lower bound on the error probability of ML detection. Leveraging this lower bound, we
  perform extensive numerical experiments demonstrating that SIW is, in most cases of interest, very close to optimal with moderate phase noise.  More
  importantly and perhaps surprisingly, such near-ML performance can be achieved by applying only twice the
  nearest neighbor detection algorithm. In this sense, our results reveal a striking fact: near-ML detection of phase noise
  corrupted MIMO channels can be done as efficiently as for conventional MIMO channels without phase noise. 
\end{abstract}


\section{Introduction}
\label{sec:introduction}

We consider the signal detection problem for the following discrete-time multiple-input multiple-output~(MIMO) channel
\begin{equation}
  \yv = \diag\bigl( e^{j\theta_{\text{r},1}}, \ldots, e^{j\theta_{\text{r}, \nr}} \bigr) \pmb{H} \diag\bigl(
  e^{j\theta_{\text{t},1}}, \ldots, e^{j\theta_{\text{t}, \nt}} \bigr) \xv + \zv, \label{eq:input-output}
\end{equation}%
where $\Hm \in \mathbb{C}^{\nr \times \nt}$ is the channel matrix known to the receiver; $\zv \in
\mathbb{C}^{\nr\times1}$ represents a realization of the additive noise whereas ${\theta_{\text{t},l}}$ and
${\theta_{\text{r},k}}$ are the phase noises at the $l$\,th transmit antenna and the $k$\,th receive antenna,
respectively; the input vector $\xv\in\mathbb{C}^{\nt\times1}$ is assumed to be carved from a quadratic amplitude modulation~(QAM). The goal is to estimate $\xv$ from the observation $\yv\in\mathbb{C}^{\nr\times1}$, with only statistical knowledge on the additive noise and the phase noises.  

In the case where the phase noise is absent, the problem is well understood, and the maximum likelihood~(ML)
solution can be found using any nearest neighbor detection~(\NND) algorithm~(see~\cite{Agrell} and the
references therein). In particular, the sphere decoder~\cite{VB} has been shown to be very efficient~\cite{HV:05a}, so that its
expected complexity~(averaged over channel realizations) is polynomial in the problem dimension $\nt$. Furthermore, there exist
approximate \NND~algorithms~(e.g., based on lattice reduction) that achieve near-ML performance when applied for MIMO detection~\cite{JE:10}. 

The presence of phase noise in \eqref{eq:input-output} is both a practical and long-standing problem in communication. In their seminal
paper~\cite{Foschini-PN-constellation} back in the 70's, Foschini \emph{et al.}~used this model to capture the
residual phase jitter at the phase-locked loop of the receiver side, and investigated both the performance of decoding algorithms as
well as constellation design in the scalar case~($\nt=\nr=1$). As a matter of fact, in most wireless
communication systems, phase noise is present due to the phase and frequency instabilities in the radio
frequency oscillators used at both the transmitter and the receiver~\cite{pn}. The
channel~\eqref{eq:input-output} can be seen as a valid mathematical model when the phase noise varies slowly as
compared to the symbol duration.\footnote{As pointed out in \cite{Ghozlan-IT} and the references therein, an
effective discrete-time channel is usually obtained from a waveform phase noise channel after filtering. When
the continuous-time phase noise varies rapidly during the symbol period, the filtered output also suffers from
amplitude perturbation. More discussion is provided in Section~\ref{sec:validity}.}
While phase noise can be practically ignored in conventional MIMO systems, its impact becomes prominent at higher carrier frequencies since it can
be shown that phase noise power increases \emph{quadratically} with carrier frequency~\cite{pn,Demir}. The
performance degradation due to phase noise becomes even more severe with the use of higher order modulations for which the angular
separation between constellation points can be small.  At medium to high SNR,  phase noise dominates additive
noise, becoming the capacity bottleneck~\cite{Durisi-capa,Yang-Shamai-PN}. As for signal detection, finding the
ML solution for the MIMO phase noise channel~\eqref{eq:input-output} is hard in general. Indeed, unlike for
conventional MIMO channels, the likelihood function of the transmitted signal cannot be obtained in closed form. 

{\bf Our Contribution.} In this work, we propose an efficient MIMO detection algorithm which finds an approximate ML
solution in the presence of phase noise. The main contributions of this work are summarized as follows.
\begin{itemize}
\item[(i)] We derive a tractable approximation of the likelihood function of the transmitted signal. While
the exact likelihood does not have a close-form expression, the proposed approximation has a simple form and
turns out to be accurate for weak to medium phase noises. 
\item[(ii)] Since maximizing the approximate likelihood function over a discrete signal set is still hard, we propose a heuristic method that finds an
approximate solution by applying twice the nearest neighbor detection algorithm. The proposed algorithm, called \emph{self-interference whitening}~(SIW), has a simple geometric interpretation. Intuitively, the phase noise perturbation generates self-interference that depends on the transmitted signal through the covariance matrix. The main idea is to first estimate the covariance of the self-interference with a potentially inaccurate
initial signal solution, then perform the whitening with the estimated covariance, followed by a second
detection. From the optimization point of view, our algorithm can be seen as a (well-chosen) concave approximation to a non-concave objective function.
\item[(iii)] We assess the performance of SIW and competing algorithms in different communication scenarios. Since the error probability of ML decoding is unknown, we propose a simulation-based lower bound which we use as a benchmark. Simulation results show that SIW achieves near ML performance in most scenarios.  In this sense, our work reveals that near optimal MIMO detection with phase noise can be done as efficiently as without phase noise. Although the likelihood approximation is derived using the assumption that the phase noise has a Gaussian distribution, our numerical experiments show that SIW works well even with non-Gaussian phase noises. 
\end{itemize}

{\bf Related Work.} Receiver design with phase noise mitigation has been extensively investigated in the past years~(see~\cite{rajet2,allerton} and references therein). More complex channel models, including multi-carrier systems~(e.g., OFDM) and time-correlated phase noises~(e.g., the Wiener process) have also been considered.  In particular, joint data detection and phase noise estimation algorithms have been proposed in \cite{mehr, rajet2}. A phase noise estimation based scheme to improve the system performance for smaller alphabets has been proposed in \cite{allerton}. However, the challenging problem of signal detection in MIMO phase noise channels using higher order modulation, where performance is extremely sensitive to phase noise, has not been addressed adequately before. 

The remainder of the paper is organized as follows. We start with a formal description of the problem in the next section. The approximation of the likelihood function is derived in Section~\ref{sec:llh}, followed by the proposed algorithm described in Section~\ref{sec:algo}. The hardness of finding the exact ML solution is investigated in Section~\ref{sec:hardness}. We present the numerical experiments in Section~\ref{sec:examples}. Further discussion on the proposed algorithm and relevance of the considered channel model is provided in Section~\ref{sec:discussions}. Section~\ref{sec:conclusions} concludes the paper.

\section{Assumptions and Problem Formulation}
\label{sec:model}

\subsection*{Notation}
Throughout the paper, we use the following notation. For random quantities, we use upper case letters,
e.g., $X$, for scalars, upper case letters with bold and non-italic fonts, e.g., $\rvVec{V}$, for vectors, and upper
case letter with bold and sans serif fonts, e.g., $\rvMat{M}$, for matrices.  Deterministic quantities are denoted in
a rather conventional way with italic letters, e.g., a scalar $x$, a vector $\pmb{v}$, and a matrix $\pmb{M}$.
The Euclidean norm of a vector $\vv$ is denoted by $\|\pmb{v}\|$,
respectively. The transpose and conjugated transpose of $\pmb{M}$ are $\pmb{M}^T$ and $\pmb{M}^H$, respectively.

\subsection{System model}

In the following, we describe formally the channel model~\eqref{eq:input-output} presented in the previous section.
We assume a MIMO channel with $\nt$ transmit and $\nr$ receive antennas.  Let $\pmb{H}$ denote the channel matrix,
where the $(k,l)$-th element of $\pmb{H}$, denoted as $h_{k,l}$, represents the channel gain between the $l$\,th
transmit antenna and $k$\,th receive antenna.  
The transmitted vector is denoted by  $\pmb{x} = [x_1, \ldots, x_{\nt}]^T$,
where $x_l\in\mathcal{X}, l = 1,\ldots,\nt$, $\mathcal{X}$ being 
typically a QAM constellation with normalized average energy, i.e., $\frac{1}{|\mathcal{X}|}
\sum_{x\in\mathcal{X}} |x|^2 = 1$. For a given transmitted vector $\xv$, the
received vector in base-band can be written as the following random vector 
\begin{equation}
  \rvVec{Y} = \rvLambdaR \pmb{H} \rvLambdaT\, \xv + \rvVec{Z},
\end{equation}
where the diagonal matrices $\rvLambdaR :=
\text{diag}\left(e^{j\Thetar{,1}}, 
\ldots,   e^{j\Thetar{,\nr}}\right)$ and $\rvLambdaT :=
\text{diag}\left(e^{j\Thetat[1]}, \ldots,   e^{j\Thetat[\nt]}\right)$
capture the phase perturbation at the receiver and transmitter, respectively; $\rvVec{Z}$ is
the additive white Gaussian noise~(AWGN) vector with $\rvVec{Z} \sim
\mathcal{CN}(0,\SNR^{-1} \Id)$, where $\SNR$ is the nominal signal-to-noise ratio~(SNR). 
We assume that the phase noise $\rvVec{\Theta} := [\Thetat[1]\ \cdots\ \Thetat[\nt]\ \Thetar[1]\ \cdots\
\Thetar[\nr]]^T$ is jointly Gaussian with $\rvVec{\Theta} \sim \mathcal{N}(0, \Qm_\theta)$ where the covariance matrix
$\Qm_\theta$ can be arbitrary. Note that this model includes as a special case the uplink channel in which $\nt$ is
the number of single-antenna users. In such a case, the transmit phase noises are independent. For simplicity, we
consider uncoded transmission in which each symbol $x_l$ can take any value from $\mathcal{X}$ with equal probability. 

Further, we assume that the channel matrix can be random but is perfectly known at
the receiver, whereas such knowledge at the transmitter side is irrelevant in
uncoded transmission. We
also define $\rvMat{H}_\Theta := \rvLambdaR \Hm \rvLambdaT$ and accordingly
$\pmb{H}_\theta$ for some realization of $\rvVec{\Theta} = \pmb{\theta}$.  By
definition, we have $\pmb{H}_{0} = \pmb{H}$. Finally, we ignore the temporal
correlation of the phase noise process and the channel process, and focus on the
spatial aspect of the signal detection problem.

\subsection{Problem formulation}

With AWGN, we have the following conditional probability density function~(pdf)
\begin{align}
p(\yv\cond \pmb{x}, \pmb{\theta}, \pmb{H})
&= \frac{\SNR^{\nr}}{\pi^{\nr}}e^{- \SNR {\| \pmb{y} - \pmb{H}_{\!\theta} \,
\pmb{x} \|^2} }, \label{eq:pdf}
\end{align}%
from which we obtain the likelihood function by integrating over $\rvVec{\Theta}$
\begin{align}
  \likelihood{\yv}{\xv}{\Hm} &= \E_{\rvVec{\Theta}} \bigl[  p(\yv\cond \pmb{x}, \rvVec{\Theta}, \pmb{H})
  \bigr] \\ &=  \ln \left( \E_{\rvVec{\Theta}} \left[  e^{- \SNR {\| \pmb{y} - \pmb{H}_{\!\rvVec{\Theta}} \, \pmb{x}
  \|^2} } \right] \right) + \ln \frac{\SNR^{\nr}}{\pi^{\nr}}. 
\end{align}%
The ML detector finds an input vector from the alphabet $\mathcal{X}^{\nt}$ such that the
likelihood function is maximized. In practice, it is often more convenient to use the log-likelihood function
as the objective function, i.e., after removing a constant term
\begin{align}
  \llh(\xv, \yv, \Hm, \SNR, \Qm_\theta) 
  :\!\!&=  \ln \left( \E_{\rvVec{\Theta}} \left[  e^{- \SNR {\| \pmb{y} - \pmb{H}_{\!\rvVec{\Theta}} \, \pmb{x}
  \|^2} } \right] \right), \label{eq:llh0}
\end{align}%
where the arguments $\SNR$ and $\Qm_\theta$ can be omitted whenever confusion is not likely. 
Thus,
\begin{align}
  \xML(\pmb{y},\pmb{H}) &:= \arg \max_{\pmb{x} \in
  \mathcal{X}^{\nt}} \llh(\xv, \yv, \Hm). \label{eq:ML}
\end{align}%
From \eqref{eq:ML} we see two main challenges to compute the optimal solution:  
\begin{enumerate}
  \item The expectation in~\eqref{eq:ML} cannot be obtained in closed form. A numerical
    implementation is equivalent to finding the numerical integral in $\nt+\nr$ dimensions. This
    can be extremely hard in high dimensions. 
  \item The size of the optimization space, $|\mathcal{X}|^{\nt}$, can be prohibitively large when the modulation size $|\mathcal{X}|$ and the input dimension~$\nt$ become large. 
\end{enumerate}
In Section~\ref{sec:hardness}, we examine in more details why both of these issues are indeed challenging.

In a conventional MIMO channel without phase noise, finding the ML solution is reduced to solving the following
problem
\begin{align}
  \xMLnaive(\pmb{y},\pmb{H}) &:= \arg \min_{\pmb{x} \in
  \mathcal{X}^{\nt} } \| \pmb{y} - \pmb{H}_0\, \pmb{x} \|^2, \label{eq:ML0}
\end{align}%
which is also called the minimum Euclidean distance detection or nearest neighbor detection~(NND).  Although the search
space in \eqref{eq:ML0} remains large, the expectation is gone. Furthermore, since the objective function is the
Euclidean distance, efficient algorithms~(e.g., sphere decoder~\cite{VB} or lattice decoder~\cite{Agrell})
exploiting the geometric structure of the problem can be applied without searching over the whole space
$\mathcal{X}^{\nt}$. It is shown in \cite{HV:05a} that the sphere decoder has a polynomial average complexity
with respect to the input dimension~$\nt$ when the channel matrix is drawn i.i.d.~from a Rayleigh distribution. 

In practice, one may simply ignore the existence of phase noise and still apply
\eqref{eq:ML0} to obtain $\hat{\pmb{x}}_{\text{ML}}^{0}$ which we refer to as the
\emph{naive ML} solution in our work. While this can work relatively well when the phase
noise is close to $0$, it becomes highly suboptimal with stronger phase noise which is usually
the case in high frequency bands with imperfect oscillators. 
In this paper, we provide a near ML solution by circumventing the two challenges mentioned earlier. We first
propose an approximation of the likelihood function. Then we propose an algorithm to solve approximately the
optimization problem~\eqref{eq:ML}.  

\section{Proposed Scheme}

\subsection{Proposed Approximation of the Likelihood Function}
\label{sec:llh}

In this section, we propose to approximate the likelihood function with the implicit assumption that the phase noise
is not large. Indeed, in practice, the standard deviation of the phase noise is
typically smaller than $10$
degrees $\approx 0.174$ rad. For stronger phase noises, it is no longer reasonable to use QAM and the problem should be addressed differently. 

The likelihood function~\eqref{eq:pdf} depends on the Euclidean norm 
$\|\yv - \LambdaR \Hm \LambdaT \xv\| = \|\LambdaR^H \yv -  \Hm \LambdaT \xv\|$, in
which the difference vector $\LambdaR^H \yv -  \Hm \LambdaT \xv$ can be rewritten
and approximated as follows
\begin{align}
  \LambdaR^H \yv -  \Hm \LambdaT \xv  &=  [-\Hm \Dm_x \ \ \Dm_y ] \begin{bmatrix} e^{j {\thetav}_{t}} \\ e^{-j{\thetav}_{r}} \end{bmatrix} \\
    &\approx  (\yv - \Hm\xv) - j [\Hm \Dm_x \ \ \Dm_y ]\,  {\thetav}, \label{eq:linapp}
\end{align}%
where we define $\Dm_x := \diag(x_1,\ldots,x_{\nt})$, $\Dm_y := \diag(y_1,\ldots,y_{\nr})$, and recall that $\thetav :=
\left[\thetav_t^T\quad \thetav_r^T\right]^T$; \eqref{eq:linapp} is from the linear
approximation\footnote{Here we use, with a slight abuse of notation,
$e^{j\thetav}$ to denote the vector obtained from the element-wise complex
exponential operation. Similarly, the little-$o$ Landau notation $o(\thetav)$ is
element-wise.} $e^{j\thetav} = 1 + j\thetav + o(\thetav)$. Thus 
the Euclidean norm has the corresponding real approximation:
\begin{align}
  \|\yv - \LambdaR \Hm \LambdaT \xv\|^2 &\approx \|\Amx {\thetav} + \bvx\|^2,
\end{align}%
where $\Amx \in \mathbb{R}^{2\nr \times (\nt+\nr)}$ and $\bvx\in \mathbb{R}^{2\nr \times 1}$ are defined as
\begin{align}
  \Amx &:= \begin{bmatrix} \Im[\Hm \Dm_x] & \Im[\Dm_y] \\ -\Re[\Hm \Dm_x] & -\Re[\Dm_y] \end{bmatrix}, \quad \bvx:= 
    \begin{bmatrix} \Re[\yv - \Hm \xv] \\ \Im[\yv - \Hm\xv] \end{bmatrix}.  \label{eq:Aandb}
\end{align}%

With the above approximation, we can derive the approximation of the log-likelihood function.

\begin{proposition}\label{prop:1}
  Let $\Amx$ and $\bvx$ be defined as in \eqref{eq:Aandb}. Then we have the following approximation of the log-likelihood function $\ln \left( \E_{\rvVec{\Theta}} \left[ e^{- \SNR {\| \pmb{y} -
  \pmb{H}_\theta \, \pmb{x} \|^2} } \right] \right)  \approx \allh(\xv,\yv,\Hm,\SNR,\Qm_\theta)$ with
  \begin{align}
    \allh(\xv,\yv,\Hm) := - \SNR\, \bvx^\T \Wmx^{-1} \bvx - \frac{1}{2} \ln\det\left( \Wmx \right),  \label{eq:aML}
  \end{align}%
  where $\Wmx$ is defined as 
  \begin{align}
    \Wmx &:= \Id + 2 \SNR \Amx \Qm_{\theta} \Amx^\T. \label{eq:Wmx}
  \end{align}%
  Hence, the proposed approximate ML~(aML) solution is 
  \begin{align}
    \xMLapp(\yv,\Hm) &:= \arg \min_{\pmb{x} \in \mathcal{X}^{\nt}} \left\{ \SNR\, \bvx^\T
    \Wmx^{-1} \bvx + \frac{1}{2} \ln\det\left( \Wmx \right) \right\}. \label{eq:MLapp} 
  \end{align}%
\end{proposition}
\begin{proof}
Since we assume that $\rvVec{\Theta} \sim \mathcal{N}(0, \Qm_\theta)$, we have
\begin{align}
  {\E_{\rvVec{\Theta}}\left[ e^{-\SNR\|\Amx \rvVec{\Theta} + \bvx \|^2} \right]}
 &= \frac{1}{\sqrt{\det(2\pi \Qm_{\theta})}}\int_{\thetav} \text{d}\,\thetav \exp\biggl( -\thetav^\T\underbrace{(\SNR\Amx^\T \Amx +
 \frac{1}{2}\Qm_\theta^{-1})}_{\frac{1}{2}\Qm^{-1}} \thetav - 2\SNR\bvx^\T \Amx\thetav -\SNR \| \bvx \|^2 \biggr) \\
 &= \sqrt{\frac{\det(\Qm)}{\det(\Qm_\theta)}} \exp\left(- \SNR\| \bvx \|^2 + \SNR^2 \bvx^\T \Amx(2\Qm)\Amx^\T \bvx\right) \nonumber \\
 &\quad \cdot \int_{\thetav} \text{d}\,\thetav \frac{1}{\sqrt{\det(2\pi\Qm)}} \exp\biggl(
 -(\thetav+\SNR(2\Qm)\Amx^\T\bvx)^\T \frac{1}{2}\Qm^{-1} (\thetav+\SNR(2\Qm)\Amx^\T\bvx)\biggr)
 \label{eq:int}
 \\
&= \sqrt{\frac{\det(\Qm)}{\det(\Qm_\theta)}} \exp\left(-\SNR \| \bvx \|^2 + \SNR^2 \bvx^\T
\Amx(2\Qm)\Amx^\T \bvx\right) \label{eq:pdfq} \\
&= \sqrt{\frac{\det(\Qm)}{\det(\Qm_\theta)}} \exp\left(- \SNR \bvx^\T (\Id - \SNR \Amx( (2\Qm_\theta)^{-1} +
\SNR\Amx^\T\Amx)^{-1}\Amx^\T )\bvx\right) \\
&= \sqrt{\frac{\det(\Qm)}{\det(\Qm_\theta)}} \exp\left(- \SNR \bvx^\T (\Id + \SNR \Amx
(2\Qm_\theta)\Amx^\T)^{-1} \bvx\right) \label{eq:inv}\\
&= \frac{1}{\sqrt{\det(\Id + 2 \SNR \Amx \Qm_\theta \Amx^\T)}} \exp\left(- \SNR \bvx^\T (\Id +
2\SNR\Amx \Qm_\theta \Amx^\T)^{-1} \bvx\right),
\end{align}%
where \eqref{eq:pdfq} holds since the integrand in \eqref{eq:int} is a pdf with respect to ${\thetav}$; \eqref{eq:inv} is from the Woodbury matrix identity $ \left(\Id+\Um\Cm\Vm \right)^{-1} = \Id -
\Um \left(\Cm^{-1}+\Vm\Um \right)^{-1} \Vm$. Taking the logarithm and we obtain the approximated
log-likelihood function~\eqref{eq:aML}. 
\end{proof}

\begin{figure}
  \begin{subfigure}[b]{0.329\textwidth}
\includegraphics[width=\textwidth]{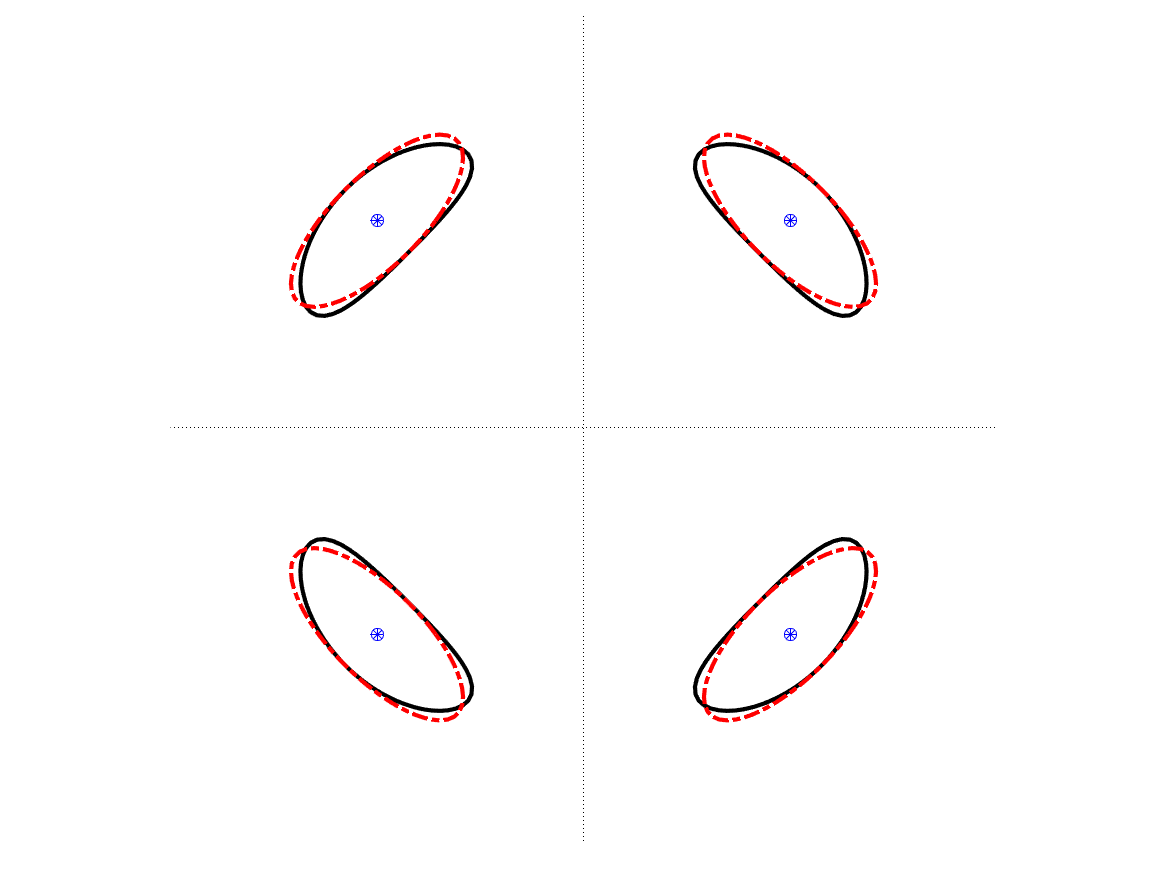}
\caption{4-QAM, $\llh = \allh = -10$.}
\label{fig:constellation}
  \end{subfigure}
  \begin{subfigure}[b]{0.329\textwidth}
\includegraphics[width=\textwidth]{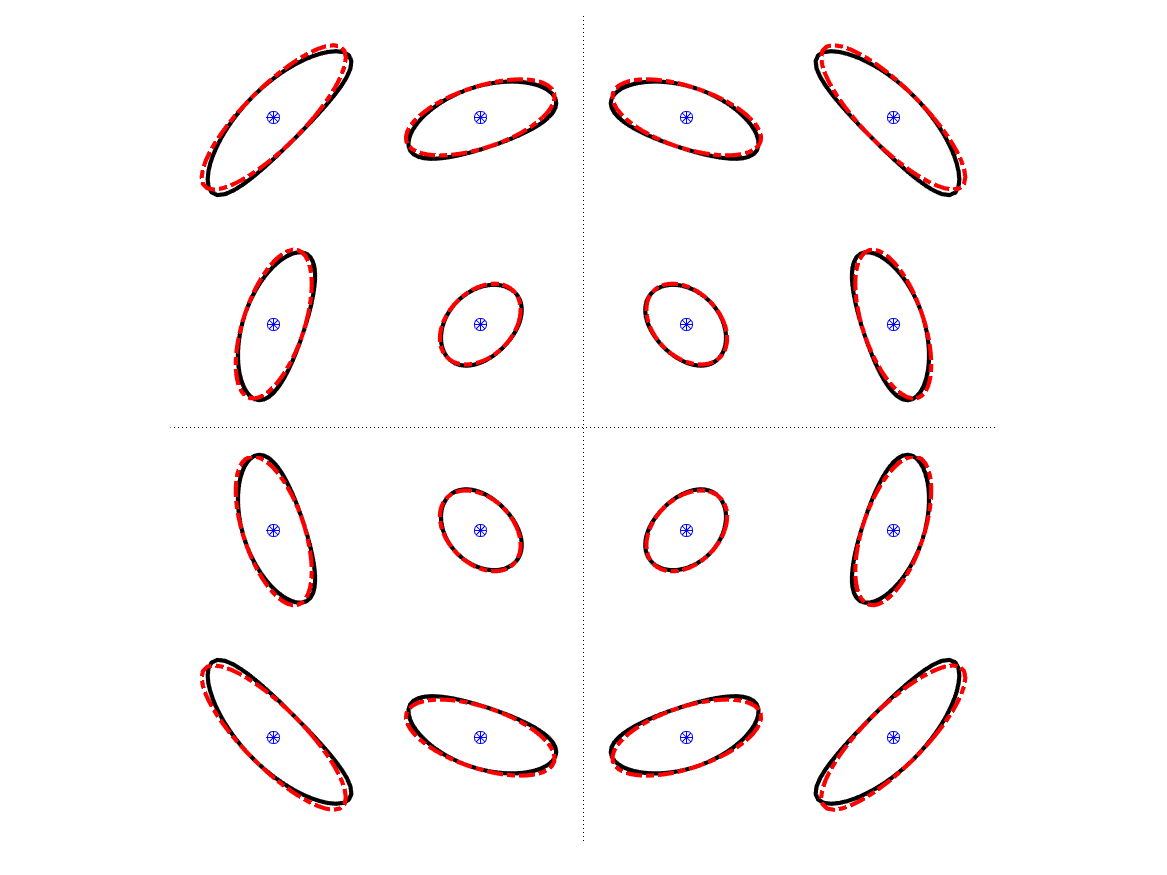}
\caption{16-QAM, $\llh = \allh = -4$.}
\label{fig:constellation16}
  \end{subfigure}
  \begin{subfigure}[b]{0.329\textwidth}
\includegraphics[width=\textwidth]{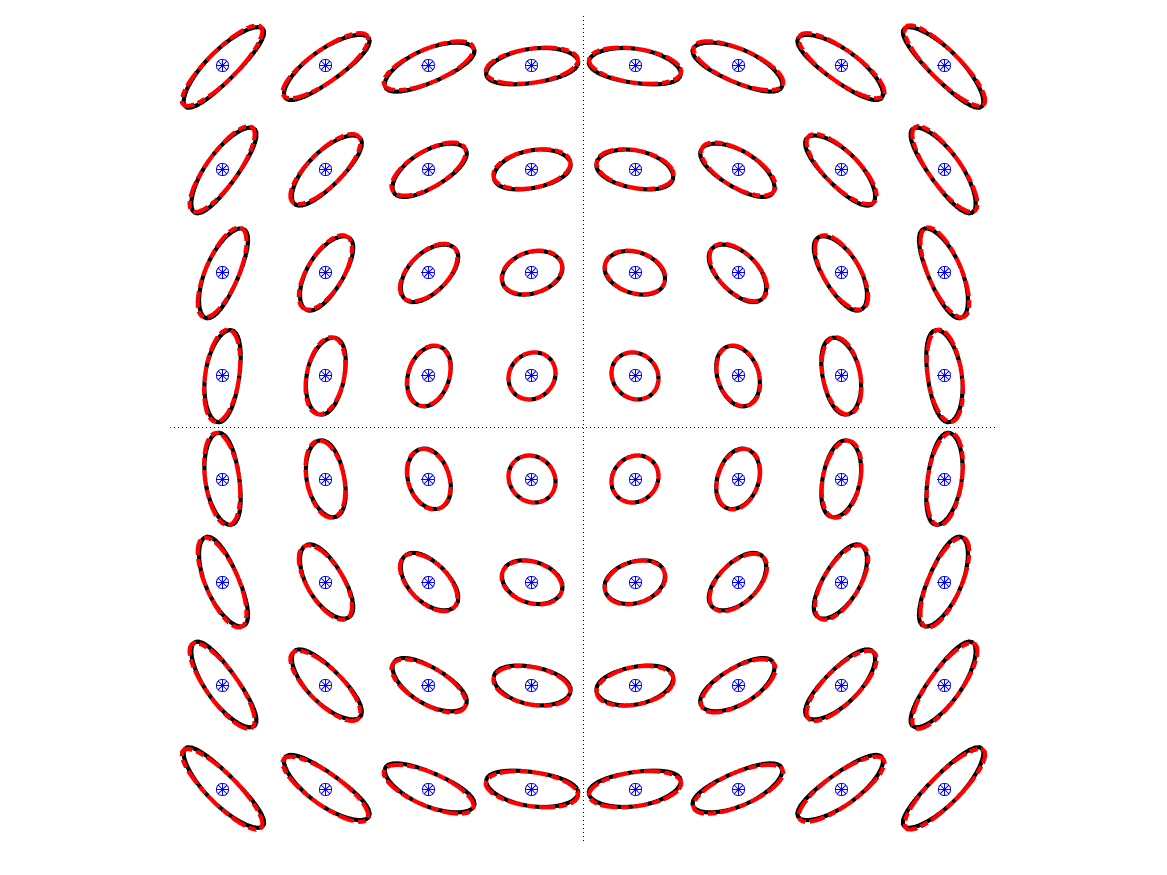}
\caption{64-QAM, $\llh = \allh = -1.6$.}
\label{fig:constellation64}
  \end{subfigure}
 \caption{The proposed approximation of the likelihood function in the scalar case. 
 Solid line is the actual likelihood level set, dashed line is the approximation. Here $\SNR=30$dB and phase noise has standard deviation~$3^\circ$ at the transmitter and at the receiver. }
  \label{fig:illustration}
\end{figure}

In Figure~\ref{fig:illustration}, we illustrate the proposed approximation for $4$-, $16$-, and $64$-QAM. In
all three cases, we plot for each constellation point a level set of the likelihood function with respect to
``$\yv$'' in solid line. The level sets of the approximated likelihood function are plotted similarly in dashed
line. While the likelihood function is evaluated using numerical integration, the approximation is in closed form given by \eqref{eq:aML}. In this figure, we observe that the approximation is quite accurate, especially
for signal points with smaller amplitude. Further, the resemblance of the level sets for the approximate
likelihood to ellipsoids suggests that the main contribution in the right hand side of \eqref{eq:aML} comes
from the first term $-\SNR\, \bvx^\T \Wmx^{-1} \bvx$. We shall exploit this feature later on to construct the
proposed algorithm.  

\begin{remark}\label{remark:1}
We can check that when $\SNR\Qm_{\theta}\to 0$, i.e., when the phase noise 
vanishes faster than the AWGN does, the above
solution minimizes $\| \bvx \|^2 = \| \yv - \Hm \xv\|^2$, which corresponds to the optimal NND
solution in the conventional MIMO case. Nevertheless, when the phase noise does not vanish, the
approximate log-likelihood function~\eqref{eq:aML} depends on $\xv$ in a rather complex way, due to
the presence of the matrix $\Wmx$. Focusing on the term $\bvx^\T \Wmx^{-1} \bvx$, we can think of
$\Wmx$ as the covariance matrix of some equivalent noise. Indeed, if we approximate the
multiplicative phase noise as an additive perturbation, then the perturbation is a self-interference
that depends on the input vector $\xv$. This perturbation is not isotropic nor circularly symmetric, and can be captured by the covariance matrix $\Wmx$. Some more
discussions in this regard will be given in the following subsection. 
\end{remark}

While the proposed approximation simplifies significantly the objective function, 
the optimization problem~\eqref{eq:MLapp} remains hard when the search space is large. For instance,
with $64$-QAM and $4\times 4$ MIMO, the number of points in $\mathcal{X}^\nt$ is more than $10^7$!
Therefore, we need further simplification by exploiting the structure of the problem.

\subsection{The Self-Interference Whitening Algorithm}
\label{sec:algo}

As mentioned above, the difficulty of the optimization~\eqref{eq:MLapp} is mainly due to the presence of the
matrix $\Wmx$ that depends on $\xv$. Let us first assume that the $\Wmx$ corresponding to the optimal
solution $\hat{\xv}_{\mathrm{aML}}$ were somehow known, and is denoted by $\Wmxhat$.  Then the optimization
problem \eqref{eq:MLapp} would be equivalent to 
\begin{align}
  \xMLapp(\yv,\Hm) &= \arg \min_{\pmb{x} \in \mathcal{X}^{\nt}} \left\{ \SNR\, \bvx^\T \Wmxhat^{-1} \bvx + \frac{1}{2} \ln\det\left( \Wmxhat \right) \right\} \\
  &= \arg \min_{\pmb{x} \in \mathcal{X}^{\nt}} \bvx^\T \Wmxhat^{-1} \bvx  \\
  &= \arg \min_{\pmb{x} \in \mathcal{X}^{\nt}}  \| \Wmxhat^{-\frac{1}{2}} \tilde{\yv} - \Wmxhat^{-\frac{1}{2}}
  \tilde{\Hm} \tilde{\xv} \|^2, 
  \label{eq:MLapp2} 
\end{align}
where $\Wmxhat^{-\frac{1}{2}}$ is any matrix such that
$\left(\Wmxhat^{-\frac{1}{2}}\right)^\H\Wmxhat^{-\frac{1}{2}} = \Wmxhat^{-1}$; 
$\tilde{\xv}$, $\tilde{\yv}$, and $\tilde{\Hm}$ are defined by
\begin{align}
  \tilde{\xv} &:= \begin{bmatrix} \Re[\xv] \\ \Im[\xv] \end{bmatrix}, \quad
    \tilde{\yv} := \begin{bmatrix} \Re[\yv] \\ \Im[\yv] \end{bmatrix}, \quad
      \tilde{\Hm} := \begin{bmatrix} \Re[\Hm] & -\Im[\Hm] \\ \Im[\Hm] & \Re[\Hm] \end{bmatrix}. 
        \label{eq:xyH}
\end{align}%
Note that for a given $\Wmxhat$, \eqref{eq:MLapp2} can be solved efficiently with any \NND~algorithm.
Unfortunately, without knowing the optimal solution $\hat{\xv}_{\mathrm{aML}}$, the exact $\Wmxhat$
cannot be found.  Therefore, the idea is to first estimate the matrix $\Wmxhat$ with some suboptimal solution
$\hat{\xv}$, and then solve the optimization problem~\eqref{eq:MLapp2} with a \NND. 
We call this two-step procedure \emph{self-interference whitening}~(SIW). 
For instance, we can use the naive ML solution $\xMLnaive$ as the initial estimate to obtain $\Wmxhat$, and have
\begin{align}
  \xMLapp'(\yv,\Hm)
  &= \arg \min_{\pmb{x} \in \mathcal{X}^{\nt}}  \| \Wmxhatnaive^{-\frac{1}{2}} \tilde{\yv} -
  \Wmxhatnaive^{-\frac{1}{2}} \tilde{\Hm} \tilde{\xv} \|^2. 
\end{align}%

\begin{remark}\label{remark:2}
  The intuition behind the SIW scheme is as follows.
  From the definition of $\Wmx$ in \eqref{eq:Wmx} and $\Amx$ in \eqref{eq:Aandb}, 
  we see that $\Wmx$ depends on $\xv$ only through $\Hm \Dm_x$. First, the column space of $\Hm
  \Dm_x$ does not vary with $\xv$ since $\Dm_x$ is diagonal. Second, a small perturbation of
  $\xv$ does not perturb $\Wmx$ too much in the Euclidean space. 
  Since the naive ML point $\xMLnaive$ is close to the actual point $\xv$ in the column space of
  $\Hm$, it provides an accurate estimate of $\Wmx$. This can also be
  observed on Figure~\ref{fig:constellation64}, where we see that the ellipsoid-like dashed lines have
  similar sizes and orientations for constellation points that are close to each other.  
\end{remark}
\begin{remark}
Another possible initial estimate is the naive linear minimum mean square error~(LMMSE) solution. As the naive
ML, the naive LMMSE ignores the phase noise and returns
\begin{align}
  \xMMSE^0(\yv, \Hm) &:= \arg\min_{\xv\in\mathcal{X}^{\nt}} \| \Hm^H (\SNR^{-1} \Id + \Hm \Hm^H)^{-1} \yv -
  \xv \|^2.  \label{eq:lmmse}
\end{align}%
It is worth mentioning that in the presence of phase noise the naive LMMSE is not necessarily dominated by the naive ML solution, as
will shown in the numerical experiments of Section~\ref{sec:examples}. 
\end{remark}

The main algorithm of this work is described in Algorithm~\ref{algo:2}. 
\begin{algorithm*}[!h]
\caption{Self-interference whitening}
\label{algo:2}
\begin{algorithmic}
  \State \underline{Input}: $\yv$, $\Hm$, $\SNR$, $\Qm_\theta$
  \State Find $\xMMSE^0$ from \eqref{eq:lmmse} 
  \State Find $\xMLnaive \gets \mathrm{\NND}(\yv, \Hm, \mathcal{X})$ 
  \If{ $\allh(\xMMSE^0,\yv,\Hm,\SNR,\Qm_\theta) > \allh(\xMLnaive,\yv,\Hm,\SNR,\Qm_\theta)$}
  \State $\hat{\xv} \gets \xMMSE^0$
  \Else
  \State $\hat{\xv} \gets \xMLnaive$
  \EndIf
  \State Generate $\Wmxhat$ from $\hat{\xv}$ using \eqref{eq:Aandb} and \eqref{eq:Wmx} 
  \State Find $\Wmxhat^{\frac{1}{2}}$ using the Cholesky decomposition
  \State Generate $\tilde{\yv}$ and $\tilde{\Hm}$ according to \eqref{eq:xyH} 
  \State $\tilde{\xv}' \gets \mathrm{real\NND}(\Wmxhat^{-\frac{1}{2}} \tilde{\yv},
  \Wmxhat^{-\frac{1}{2}} \tilde{\Hm} \tilde{\xv}, \tilde{\mathcal{X}})$
  \State $\hat{\xv}' \gets \mathrm{complex}(\tilde{\xv}')$
  \If{ $\allh(\hat{\xv}',\yv,\Hm,\SNR,\Qm_\theta) > \allh(\hat{\xv},\yv,\Hm,\SNR,\Qm_\theta)$}
  \State {$\xMLapp' \gets \hat{\xv}'$}
  \Else
  \State {$\xMLapp' \gets \hat{\xv}$}
  \EndIf
  \State \underline{Output}: $\xMLapp^{(2)}$
\end{algorithmic}
\end{algorithm*}
In the algorithm, the complex function $\mathrm{\NND}({\yv}, {\Hm},
{\mathcal{X}})$ finds among the points from the alphabet ${\mathcal{X}}$ the closest one
to ${\yv}$ in the column space of ${\Hm}$; the  function $\mathrm{real\NND}(\tilde{\yv},
\tilde{\Hm}, \tilde{\mathcal{X}})$ is the real counterpart of $\mathrm{\NND}$. The function
``$\mathrm{complex}(\tilde{\xv}')$'' embeds the real vector $\tilde{\xv}'$ to the complex space by taking the upper half as the real part and the lower half as the imaginary part. 
It is worth noting that the newly obtained point is accepted only when it has a higher approximate likelihood
value than the naive ML point does. An example of the scalar case is provided in Figure~\ref{fig:example} where
$256$-QAM is used. The transmitted point is $x$ and the received point is $y$. The solid line is the level set
of the likelihood function. If the likelihood function was computed for each point in the constellation, then
one would recover $x$ from $y$ successfully. But this would be hard computationally. With the Euclidean
detection, $\hat{x}$ that is closer to $y$ than $x$ is would be found instead, which would cause an erroneous
detection. The SIW algorithm can ``correct'' the error as follows. First, to estimate the unknown matrix $\Wmx$, we compute the matrix $\Wmxhat$
which is represented by the red dashed ellipse around $\hat{x}$. 
We can see that the estimate $\Wmxhat$ is very close to the correct value $\Wmx$, given by the actual $x$ (blue dashed line). Then, we generate the coordinate system with $\Wmxhat$ and search for the closest constellation point to
$y$ in this coordinate system. In this example, $x$ can be recovered successfully. More importantly,
computationally efficient NND algorithms can be used to perform the search.

\begin{remark}
  The complexity of the SIW algorithm is essentially twice that of the \NND~algorithm used, since the other
  operations including the LMMSE detection have at most cubic complexity with respect to the dimension of the
  channel. The complexity of the \NND~algorithm depends directly on the conditioning of the given matrix. If
  the columns are close to orthogonal, then channel inversion is almost optimal. However, in the worse case, when the matrix is
  ill-conditioned, the \NND~algorithm can be slow and its complexity is exponential in the problem dimension. As mentioned earlier,
  there exist approximate~\NND~algorithms, e.g., based on lattice reduction, that can achieve near optimal performance with much lower complexity. 
\end{remark}

\begin{figure}
  \centering
\includegraphics[width=0.7\textwidth]{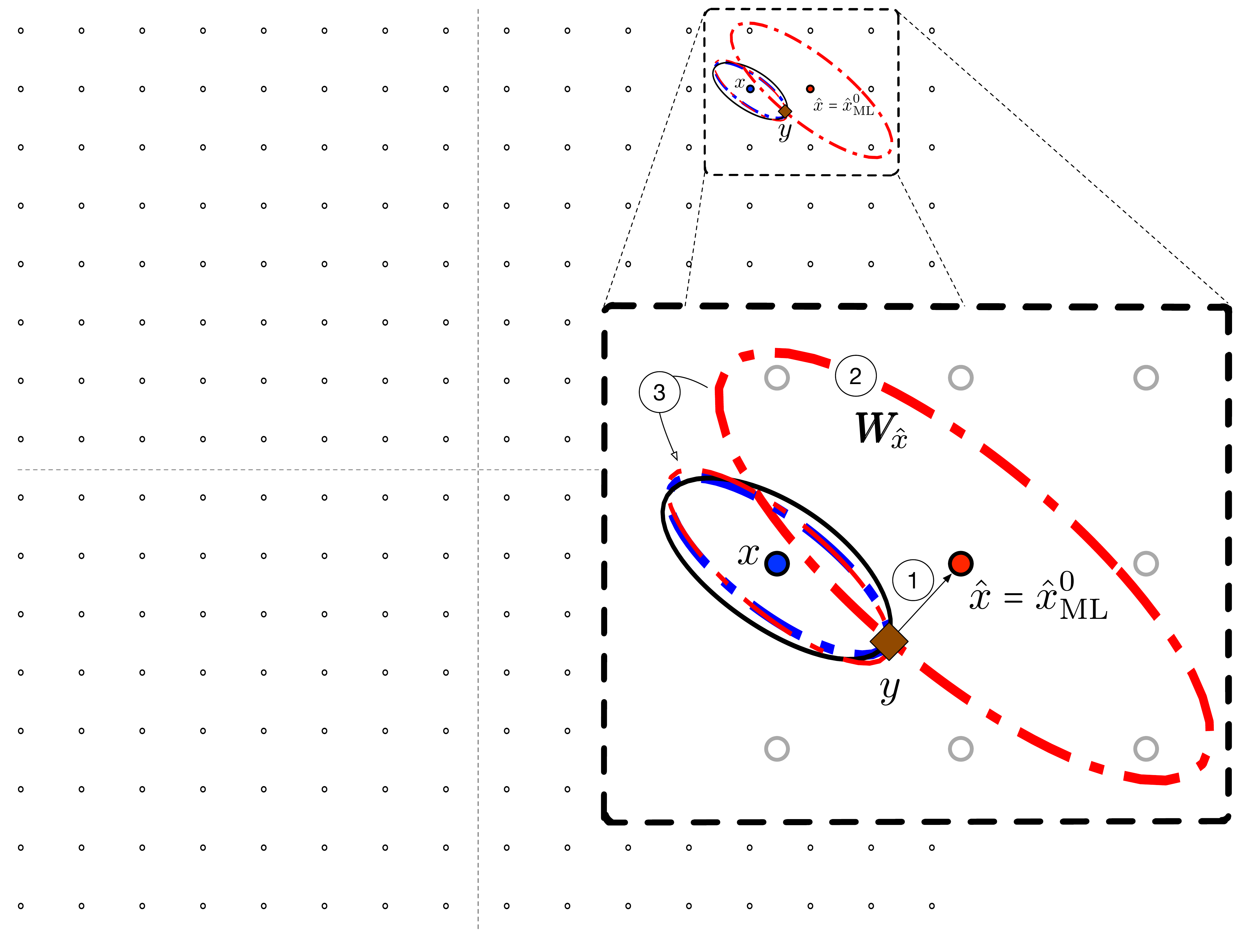}
\caption{Illustration of the proposed detection in the scalar case. An example with 256-QAM, PN~$2^\circ$. The dashed
lines represent the ellipse defined by the matrix $\Wm_{\!\!\hat{x}}$~(in red) and $\Wm_{\!\!{x}}$~(in blue). }
\label{fig:example}
\end{figure}

\section{Hardness of ML Decoding}
\label{sec:hardness}

	In the section we explore how ML decoding may be implemented. While in most cases of interest the SIW algorithm gives near-optimal performance, and ML decoding is too costly computationally, it is useful to simply provide lower bounds on the performance of any decoding algorithm. We comment on the difficulty of implementing ML decoding, in particular when the dimensions $\nr,\nt$ are large. 
        For simplicity, we assume that $\{\rV{\Theta}_{\text{t},k}\}_{k=1,\ldots,\nt}$ are
        i.i.d.~$\mathcal{N}(0,\sigma_t^2)$ and $\{\rV{\Theta}_{\text{r},k}\}_{k=1,\ldots,\nr}$ are
        i.i.d.~$\mathcal{N}(0,\sigma_r^2)$.

\subsection{Hardness of computing the likelihood} 

To compute the likelihood one needs to compute $\E_{\rvVec{\Theta}} \left[  e^{- \SNR {\| \pmb{y} - \pmb{H}_{\!\rvVec{\Theta}} \, \pmb{x} \|^2} } \right]$. For large dimensions, this seems impossible to do in closed form or using numerical integration. However a natural alternative is to use the Monte-Carlo method, using the following estimate:
 \begin{equation}
 \hat \llh(\xv, \yv, \Hm, \SNR, \Qm_\theta) \approx \hat{\rV{F}}_s := \ln\left( {1 \over s} \sum_{t=1}^s e^{- \SNR {\| \pmb{y} - \pmb{H}_{\!\rvVec{\Theta}^{(t)}} \, \pmb{x} \|^2} } \right),
 \end{equation}
where $\rvVec{\Theta}^{(1)},\dots,\rvVec{\Theta}^{(s)}$ are $s$ i.i.d.~copies of $\rvVec{\Theta}$. 
Using the delta method~\cite{Oehlert92} and the central limit theorem, the asymptotic variance of this estimator is 
\begin{equation}
  \Var( \hat{\rV{F}}_s) \sim {1 \over s} {\Var\bigl(e^{- \SNR {\| \pmb{y} -
\pmb{H}_{\!\rvVec{\Theta}} \pmb{x} \|^2}}\bigr) \over \Bigl( \E[e^{- \SNR {\| \pmb{y} -
\pmb{H}_{\!\rvVec{\Theta}} \pmb{x} \|^2} }]\Bigr)^2}   \; , \quad s \to \infty. \label{eq:asymp_var} 
\end{equation}
Let us compute this quantity in the (arguably easiest) case where $\Hm$ is the identity matrix
with $\nt=\nr=n$:
\begin{equation}
	\| \pmb{y} - \pmb{H}_{\!\rvVec{\Theta}} \, \pmb{x} \|^2 = \sum_{k=1}^{n} \left|y_k -
        x_{k} e^{j(\rV{\Theta}_{r,k} + \rV{\Theta}_{t,k})}\right|^2.  
\end{equation}
	Define the one-dimensional likelihood 
        \begin{equation}
        g(x,y,\SNR,\sigma) := \E[e^{-\SNR|y - x e^{j(\rV{\Theta}_{r} + \rV{\Theta}_{t})}|^2}],
        \end{equation}
	and using independence we calculate the moments:
        \begin{align}
		\E[e^{- \SNR {\| \pmb{y} - \pmb{H}_{\!\rvVec{\Theta}} \, \pmb{x} \|^2} }] &=
                \prod_{k=1}^{n} g(x_k,y_k,\SNR,\sigma), \\
		\Var(e^{- \SNR {\| \pmb{y} - \pmb{H}_{\!\rvVec{\Theta}} \, \pmb{x} \|^2} }) &= \prod_{k=1}^{n} g(x_k,y_k,2\SNR,\sigma) - \prod_{k=1}^{n} g(x_k,y_k,\SNR,\sigma)^2.
        \end{align}%
        From \eqref{eq:asymp_var}, the asmyptotic variance is hence:
        \begin{align}
          \Var( \hat{\rV{F}}_s)  =  {1 \over s} \left(\prod_{k=1}^{n} v(x_k,y_k,\SNR,\sigma) -
          1\right). \label{eq:asymp_var2} 
        \end{align}%
         where $v(x,y,\SNR,\sigma) := {g(x,y,2\SNR,\sigma) \over g(x,y,\SNR,\sigma)^2}$. It is
         noted that, unless $\sigma^2_r + \sigma^2_t = 0$ or $x = 0$, or $y = 0$, the random variable $e^{-\SNR|y - x
         e^{j(\rV{\Theta}_{r} + \rV{\Theta}_{t})}|^2}$ cannot be a constant.
	Thus, we have $v(x_k,y_k,\SNR,\sigma) > 1$ for $k=1,\ldots,n$. As a result, the asymptotic
        error \eqref{eq:asymp_var2} grows exponentially with the dimension $n$, so that the
        Monte-Carlo method is infeasible in high dimensions, since the number of samples must
        also scale exponentially with $n$ to maintain a constant error. 
	
        For instance, consider $x_k=y_k=1$ for all $k$, $\SNR = 40 \textrm{dB}$, 
        $\nr=\nt=20$, and $\sigma_r = \sigma_t = 3$ degrees. Then $v \approx 7$,
        so that, to obtain an error smaller than $0.1$, one would require $s \approx 10^{18}$
        samples, which is clearly not feasible in practice. 
	
\subsection{Hardness of maximum likelihood search} 

Assume that one is able to estimate the value of the likelihood with high accuracy (as
seen above this is typically hard), and denote by $\bar{\llh}$ this value. We may then
consider the following algorithm. Given the received symbol $\yv$, and a radius $\rho$, one first
computes the set of points $\mathcal{S} := \{\pmb{x} \in \mathcal{X}^{\nt}:  \| \pmb{y} -
\pmb{H} \pmb{x} \|^2 \le \rho^2\}$, then one computes the value of $\bar{\llh}$ for each of those points and returns the point maximizing $\bar{\llh}$. In the large system limit, the following concentration phenomenon occurs.

\begin{proposition}\label{th:radius}
	Assume that $\rvMat{H}$ has i.i.d.~entries with distribution $\mathcal{CN}(0,1)$ and
        $\rvVec{X}$ is chosen uniformly at random from ${\mathcal X}^{\nt}$. Define the radius
        $\rV{R}^2 :=  \| \rvVec{Y}   -  \rvMat{H} \rvVec{X} \|^2$. Then we have:
	\begin{equation}
          \E[\rV{R}^2] = 2 \nt \nr \Bigl(1 - e^{-{\sigma_r^2 + \sigma_t^2 \over 2}} \Bigr) +
          \SNR^{-1}\,\nr,
	\end{equation}
	and for any $\eta > 0$ we have:
	\begin{equation}
          \PP\Bigl\{ (1-\eta) \E[\rV{R}^2] \le \rV{R}^2 \le (1+\eta) \E[\rV{R}^2]\Bigr\} \to 1,
          \quad  \nt,\nr \to \infty.
	\end{equation}
\end{proposition}
Note that the above result is general and does not impose that the number of antennas $\nr,\nt$ scale at the same speed. We draw two conclusions from this result: (i) Any such algorithm applied with radius $(1 + \eta) \sqrt{\E[R^2]}$ for any $\eta > 0$ is guaranteed to inspect the optimal point with high probability. (ii) Any such algorithm which has a large success probability needs to inspect every point in a sphere of radius ${O}\left( \sqrt{\nt \nr} \sqrt{1 - \exp\left(-{\sigma_r^2 + \sigma_t^2 \over 2}\right)}\right)$, and therefore typically has a very high complexity. 

From the above analysis we see that the second difficulty of the decoding problem for phase noise channels, even when the likeihood can be computed, lies in the number of points to be inspected which is exponentially large in $\nr,\nt$. The problem is that computing the likelihood at any given point does not give us any information about the structure of $\llh$, and does not help in maximizing $\llh$ efficiently. 

\subsection{Non-concavity in the high SNR regime} 

Indeed maximizing $\llh$ is difficult, and it seems that even performing zero-forcing, i.e.,
maximizing $\llh$ over $x \in \mathbb{C}^{\nt}$ rather than over $\xv \in {\cal X}^{\nt}$, is
difficult since $\llh$ is non-concave, at least in the high SNR regime.  We first show that, in
the high SNR regime, the log-likelihood can be approximated by a function of the minimal value of
$\| \pmb{y} - \pmb{H}_{\!{\thetav}}\, \pmb{x} \|^2$ , where the minimum is taken over all possible
phases ${\thetav}$.
\begin{lemma}\label{prop:highsnr}
	For all $\xv ,\yv, \Hm, \Qm_\theta$, we have the following high SNR behavior
        \begin{equation}
		\lim_{\SNR \to \infty} -{1 \over \SNR} \llh(\xv ,\yv, \Hm, \SNR, \Qm_\theta) =
                m(\xv ,\yv, \Hm) :=  \min_{\thetav \in \mathbb{R}^{\nt + \nr}} \| \pmb{y} -
                \pmb{H}_{\!\thetav}\, \pmb{x} \|^2.
        \end{equation}%
\end{lemma}
\begin{proof}
Consider $\epsilon > 0$, we have that
\begin{equation}
  \indic\left\{  \| \pmb{y} - \pmb{H}_{\!\rvVec{\Theta}} \, \pmb{x} \|^2 \le m(\xv ,\yv, \Hm) +
  \epsilon \right\} e^{- \SNR (m(\xv ,\yv, \Hm) + \epsilon)}   \le e^{- \SNR \| \pmb{y} - \pmb{H}_{\!\rvVec{\Theta}} \pmb{x} \|^2} \le e^{- \SNR m(\xv ,\yv, \Hm)}.
\end{equation}
	Taking expectations and then logarithms:
        \begin{equation*}
		-(m(\xv ,\yv, \Hm)+\epsilon) + {1 \over \SNR} \ln \PP\Bigl\{\| \pmb{y} -
                \pmb{H}_{\!\rvVec{\Theta}} \pmb{x} \|^2 \le m(\xv ,\yv, \Hm) + \epsilon\Bigr\} \le
                {1 \over \SNR}  \llh(\xv ,\yv, \Hm, \SNR, \Qm_\theta) \le -m(\xv ,\yv, \Hm).
        \end{equation*}
	Since the mapping ${\thetav} \mapsto \| \pmb{y} - \pmb{H}_{\!{\thetav}}\, \pmb{x} \|^2$ is
        continuous, for any given $\epsilon>0$, the probability \newline $\PP\bigl\{\| \pmb{y} -
        \pmb{H}_{\!{\rvVec{\Theta}}}\, \pmb{x} \|^2 \le m(\xv ,\yv, \Hm) + \epsilon\bigr\}$ is
        strictly positive. Letting $\SNR \to \infty$, we have
	$$
		-(m(\xv ,\yv, \Hm)+\epsilon) \le \lim_{\SNR \to \infty} \inf {1 \over \SNR}  \llh(\xv ,\yv, \Hm, \SNR, \Qm_\theta) \le \lim_{\SNR \to \infty} \sup {1 \over \SNR}  \llh(\xv ,\yv, \Hm, \SNR, \Qm_\theta) \le -m(\xv ,\yv, \Hm).
	$$
	Since the above holds for all $\epsilon>0$, we have $\lim_{\SNR \to \infty} -{1 \over \SNR} \llh(\xv ,\yv, \Hm, \SNR, \Qm_\theta) = m(\xv ,\yv, \Hm)$.
\end{proof}

We can now show that in general the log-likelihood is not concave, hence maximizing it is not straightforward.  
\begin{proposition}\label{prop:non_concave}
For $\SNR$ large enough and $\Hm\ne\zerov$, there exists $\yv$ such that $\xv \mapsto \llh(\xv, \yv, \Hm, \SNR, \Qm_\theta)$ is a non-concave function.
\end{proposition}
\begin{proof}
Assume that $\llh$ is concave, then for all $\xv$ we must have:
\begin{equation}
	{1 \over 2}\bigl(\llh(\xv ,\yv, \Hm, \SNR, \Qm_\theta) + \llh(\xv^* , \yv, \Hm,
        \SNR,\Qm_\theta)\bigr) \le \llh\left({\xv + \xv^{*} \over 2}, \yv, \Hm, \SNR, \Qm_\theta\right).
\end{equation}
From Lemma~\ref{prop:highsnr}, the above inequality implies that 
\begin{equation}
	{1 \over 2} \bigl(m(\xv ,\yv, \Hm) + m(\xv^* ,\yv, \Hm) \bigr) \ge m\left({\xv + \xv^{*} \over 2},
        \yv, \Hm\right). \label{eq:tmp898} 
\end{equation}
We shall construct an example to show that the above does not hold in general.
Consider any $\zv$ such that $\Hm \zv \ne \zerov$, and let $\yv = \Hm\zv$ and $\xv = j(|z_1|,...,|z_{\nt}|)$, so that $\xv^* = -\xv$. Then $\xv$
and $\zv$ are equal up to a phase transformation, so $m(\xv ,\yv, \Hm) = m(\zv ,\yv, \Hm) = 0$.
Similarly $m(\xv^* ,\yv, \Hm) = 0$. By definition $m({\xv + \xv^{*} \over 2}, \yv, \Hm) =
m(\zerov, \yv, \Hm) = \| \yv \|^2$. In this example, \eqref{eq:tmp898} would imply $0 \ge \| \yv
\|^2$, which is clearly a contradiction since $\yv \ne \zerov$. Hence $\llh$ cannot be concave for $\SNR$ large enough. 
\end{proof}

This fact that the log-likelihood is in general non-concave gives another important insight: SIW
can in fact be seen as a (well chosen) concave approximation of a non-concave function. To circumvent
the problem of non-concavity of the log-likelihood SIW approximates it by a function which, when
$W_{\xv}$ is fixed, is concave. As non-concavity appears mainly for high SNRs, the discrepancy
between the performance of ML and SIW (if any) should be more visible in the high SNR regime, and
this will be confirmed by our numerical experiments in the next section.

\section{Numerical Experiments}
\label{sec:examples}

In this section, we look at different communication scenarios in which we compare the proposed scheme to 
some baseline schemes, including two schemes that ignore the phase noise:
\begin{itemize}
  \item The naive LMMSE solution given by \eqref{eq:lmmse}
  \item The naive ML solution~\eqref{eq:ML0}
\end{itemize}
and a scheme that takes the phase noise into account:
\begin{itemize}
  \item Selection between naive LMMSE and ML: the receiver first finds the naive LMMSE and naive ML solutions, then computes the proposed approximate likelihood function and selects the one with higher value. 
\end{itemize}
Note that we focus on the vector detection error rate\footnote{The detection is considered successful only when
all the symbols in $\xv$ are recovered correctly. Otherwise an error is declared.} as our performance metric.

\subsection{Simulation-based lower bounds}

\newcommand{\PeML}{P_{\text{e}}^{\mathrm{ML}}}
\newcommand{\PeMLapp}{P_{\text{e}}^{\mathrm{aML}}}
In order to appreciate the performance of the proposed algorithm, we need to compare it not only with the
existing schemes, but also to the fundamental limit given by ML detection, which is optimal. 
Let us recall that the proposed SIW algorithm may suffer from two levels of suboptimality. First, 
the approximate likelihood function \eqref{eq:aML} may be inaccurate in some cases. Second, even if
\eqref{eq:aML} is accurate, the SIW algorithm is not guaranteed to find the optimal solution \eqref{eq:MLapp}.
Therefore, in order to identify the source of the potential suboptimality, it would be useful to compare the
performance of the SIW scheme with the performance given by \eqref{eq:aML} and with that given by \eqref{eq:ML}. 

Unfortunately, as pointed out earlier, finding \eqref{eq:aML} requires an exhaustive search with complexity
growing as $|\mathcal{X}|^{\nt}$. Finding \eqref{eq:ML} is even harder because of the numerical
multi-dimensional integration.
As such, we resort to lower bounds on the detection error for \eqref{eq:aML} and \eqref{eq:ML}, respectively,
which are enough for our purpose of benchmarking. To that end, we write
\begin{align}
  \PeMLapp &\ge \mathbb{P}\left\{ \allh(\rvVec{X}, \rvVec{Y}, \rvMat{H}) < \max_{\xv\in\mathcal{X}^{\nt}} \allh(\xv, \rvVec{Y},  \rvMat{H}) \right\} \label{eq:aMLLB0} \\
  &\ge \mathbb{P}\left\{ \allh(\rvVec{X}, \rvVec{Y}, \rvMat{H}) < \max_{\xv\in\mathcal{L}} \allh(\xv,
  \rvVec{Y},  \rvMat{H}) \right\} \label{eq:aMLLB},\quad \forall\,\mathcal{L}\subseteq \mathcal{X}^{\nt},
\end{align}%
where the first lower bound is from the definition of the detection criterion, namely, error occurs if there
exists at least one input vector that has a strictly higher approximate likelihood value. Note that although
the second inequality is valid for all $\mathcal{L}$, it becomes equality if $\mathcal{L}$ contains \emph{all}
the points in $\mathcal{X}^{\nt}$ that have a higher approximate likelihood value than $\rvVec{X}$ does. In
this work, we only take a large set around $\rvVec{X}$ to obtain the lower bound~\eqref{eq:aMLLB}, without any
theoretical guarantee of tightness of \eqref{eq:aMLLB}. Similarly, for ML detection, we have  
\begin{align}
  \PeML &\ge \mathbb{P}\left\{ \llh(\rvVec{X}, \rvVec{Y}, \rvMat{H}) < \max_{\xv\in\mathcal{X}^{\nt}} \llh(\xv, \rvVec{Y},  \rvMat{H})
  \right\} \label{eq:MLLB0} \\
  &\ge \mathbb{P}\left\{ \llh(\rvVec{X}, \rvVec{Y}, \rvMat{H}) < \llh(\rvVec{X}', \rvVec{Y},  \rvMat{H}) \right\}
  \label{eq:MLLB} \\
  &= \mathbb{P}\left\{ \rvVec{X} \ne \rvVec{X}', \  \llh(\rvVec{X}, \rvVec{Y}, \rvMat{H}) < \llh(\rvVec{X}', \rvVec{Y},  \rvMat{H})
  \right\}, \quad\forall\,\rvVec{X}'\in\mathcal{X}^{\nt}, \label{eq:MLLB2}
\end{align}%
where the first lower bound is from the definition of the ML detection criterion, namely, error occurs if there
exists at least one input vector that has a strictly higher likelihood value; the equality
\eqref{eq:MLLB2} holds since $\rvVec{X} \ne \rvVec{X}'$ is a consequence of $\llh(\rvVec{X}, \rvVec{Y}, \rvMat{H}) <
\llh(\rvVec{X}', \rvVec{Y},  \rvMat{H})$;
Note that the second
lower bound \eqref{eq:MLLB} holds for any vector $\rvVec{X}'$ from the alphabet $\mathcal{X}^{\nt}$ and with equality when
$\rvVec{X}'$ is the exact ML solution. Since the ML solution is unknown, we can use any suboptimal solution instead and still
obtain a valid lower bound. Now we can see that \eqref{eq:MLLB2} is much easier to evaluate than \eqref{eq:MLLB0} is since
there is no need to perform the maximization over $\mathcal{X}^{\nt}$. Intuitively, if $\rvVec{X}'$ is a near ML solution, then the lower bound should be tight enough. We shall have some more discussions on this with the upcoming numerical examples. The lower bound~\eqref{eq:MLLB2} can be obtained by simulation: 
\begin{enumerate}
  \item For a given observation $\yv$ and channel $\Hm$, find a suboptimal solution $\xv'$. 
  \item If $\xv' \ne \xv$, compute $\llh(\xv,\yv,\Hm)$ and $\llh(\xv',\yv,\Hm)$, count an ML error only when $\llh(\xv,\yv,\Hm)<\llh(\xv',\yv,\Hm)$; otherwise the counter remains unchanged.  
\end{enumerate}
With the proposed method, we need to perform twice the numerical integration~(e.g., Monte-Carlo
integration) only when $\xv'\ne\xv$. If the latter event happens with small probability, then the
average complexity to evaluate \eqref{eq:MLLB2} is low. In other words, using a $\xv'$ from a
better reference scheme not only makes the lower bound tighter but also makes it easier to
evaluate.

\subsection{Scenario~1: Point-to-point SISO channel} 

The first scenario focuses on the point-to-point Rayleigh fading single-antenna channels, also known as
single-input single-output~(SISO) channels. We consider three different modulation orders ($64$, $256$, and
$1024$) with correspondingly three values of phase noise strength~($3^\circ$, $2^\circ$, and $1^\circ$) in
terms of the standard deviation at both the transmitter and receiver sides. The idea is to assess the
performance of the proposed algorithm in different phase noise limited regimes. In the SISO case, we compare
the proposed scheme with the naive ML scheme which consists in a simple threshold detection for the real and
imaginary parts. Several remarks are in order. First, we see that ignoring the existance of phase noise incurs a significant performance loss. Second, if exhaustive search is done with the proposed likelihood approximation, then it achieves the ML performance. This can
be seen from the fact that the proposed simulation-based lower bound overlaps with the curve with exhaustive
search. This confirms the accuracy of the closed-form approximation \eqref{eq:aML} at least in the SISO case.
It is worth mentioning that the exact likelihood in this case can also be derived via the Tikhonov distribution
as shown in~\cite{Foschini-PN-constellation, Caire_PN} where the analytic expression involving the Bessel function has been provided.
Finally, more remarkably, the SIW algorithm almost achieves the ML performance without exhaustive search.

\begin{figure}[t]
  \begin{subfigure}[b]{0.329\textwidth}
    \includegraphics[width=\textwidth]{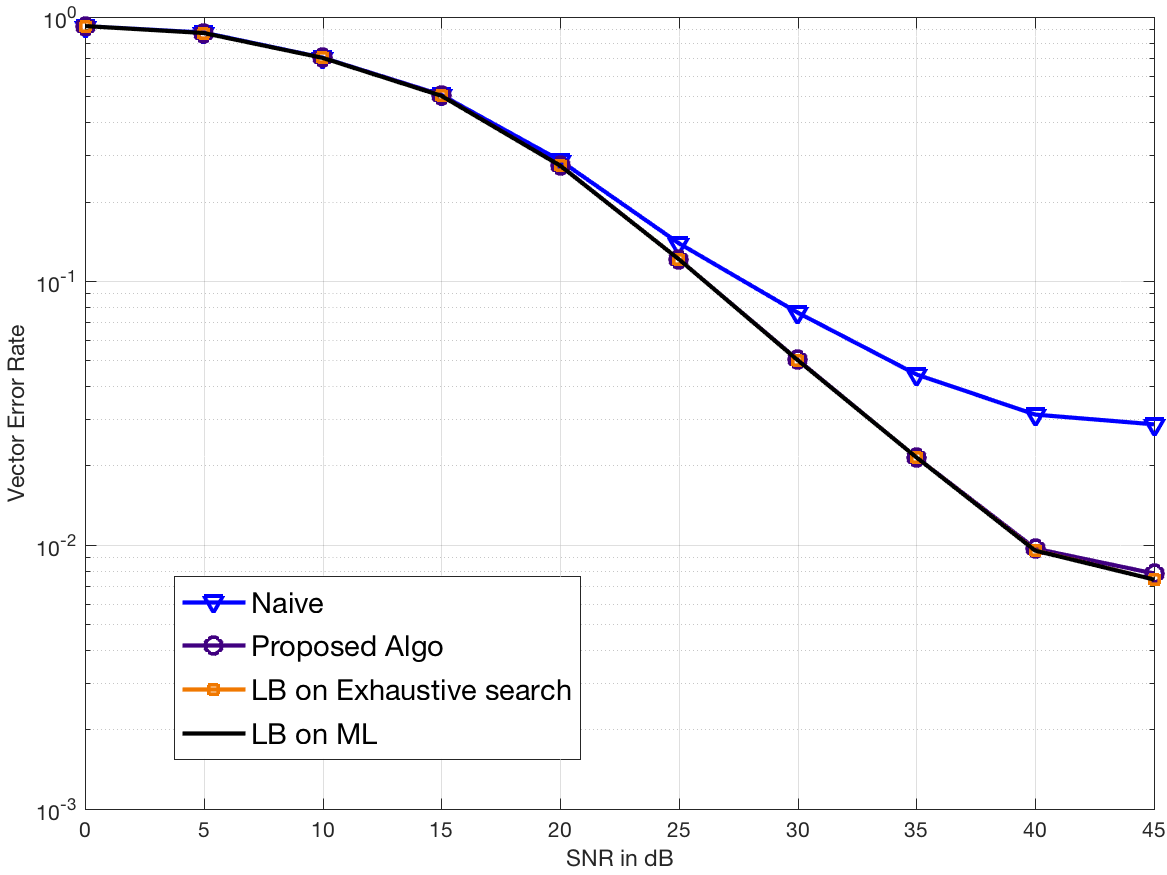}
    \caption{$64$-QAM, PN $3^\circ$.}
    \label{fig:SISO64}
  \end{subfigure}
  \begin{subfigure}[b]{0.329\textwidth}
    \includegraphics[width=\textwidth]{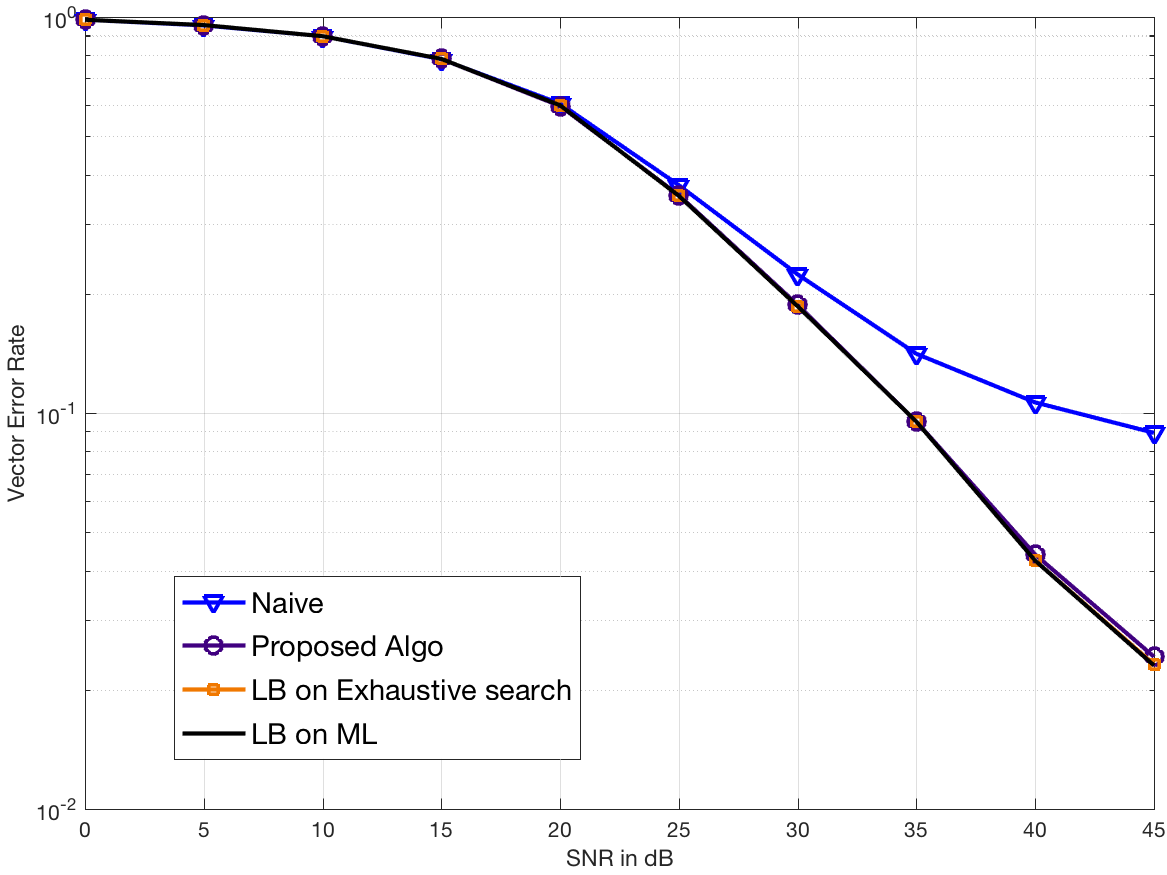}
    \caption{$256$-QAM, PN $2^\circ$.}
    \label{fig:SISO256}
  \end{subfigure}
  \begin{subfigure}[b]{0.329\textwidth}
    \includegraphics[width=\textwidth]{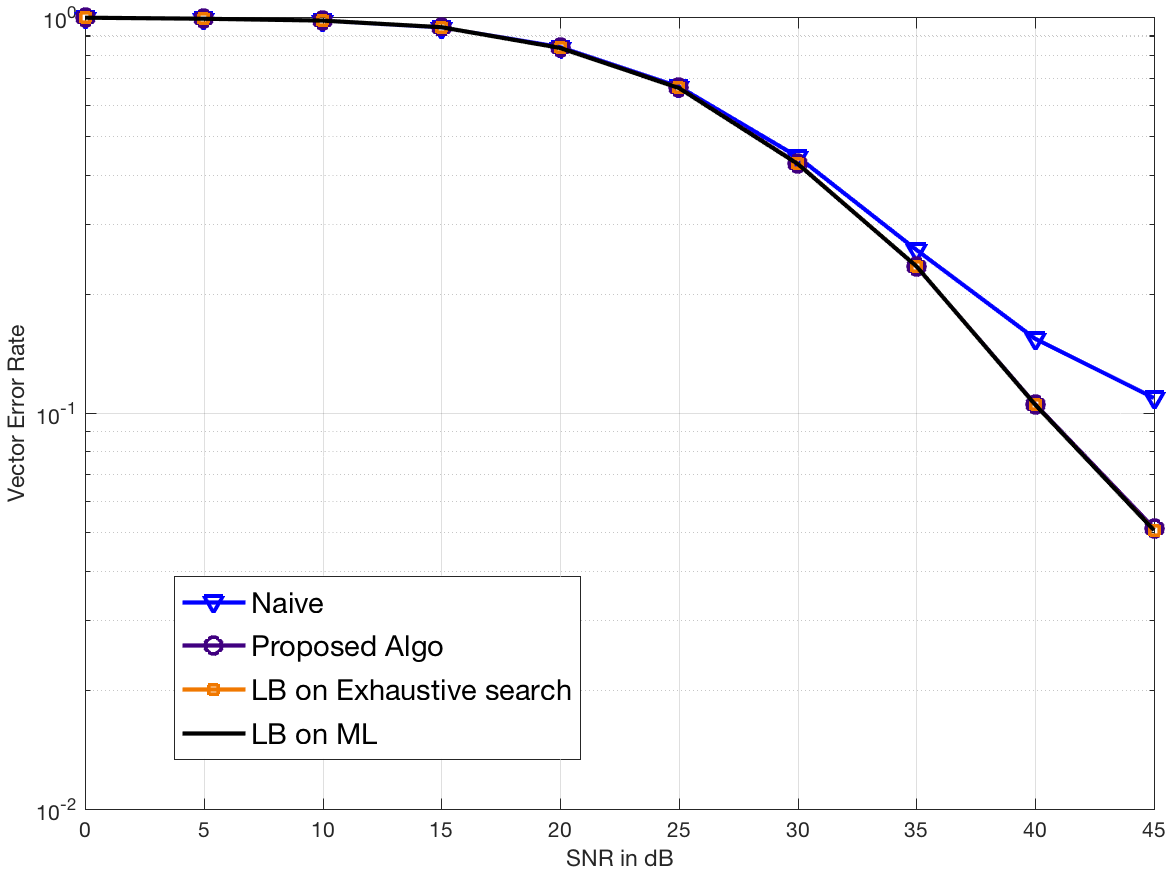}
    \caption{$1024$-QAM, PN $1^\circ$.}
    \label{fig:SISO1024}
  \end{subfigure}
  \caption{SISO Rayleigh fading with i.i.d.~phase noise.}
  \label{fig:SISO}
\end{figure}

\subsection{Scenario~2: Point-to-point LoS-MIMO channel} 

The second scenario is the point-to-point line-of-sight~(LoS) MIMO system commonly deployed as microwave
backhaul links~\cite{Bohagen,Durisi-capa,ferrand2016blind}. We assume that the channel is constant over time
but each antenna is driven by its own oscillator. This is the worst-case assumption but also often motivated by
the fact that the communication distance is large and thus the distance between antenna elements is increased
accordingly to make sure that the channel matrix is well conditioned~\cite{Bohagen,ferrand2016blind}. Here, we
adopt the model with two transmit and two receive antennas each one with dual polarizations. This is
effectively a $4\times4$ MIMO channel. The optimal distance between the antenna elements at each side can be
derived as a function of the communication distance~\cite{Bohagen}. However, it may not always be possible to
install the antennas with the optimal spacing due to practical constraints. The condition number of the channel
matrix increases when the antenna spacing decreases away from the optimal distance. In
Figure~\ref{fig:LoSMIMO}, we consider three configurations with distances $0.33$, $0.7$, and $1$ of the optimal
value, generated using the model from \cite{ferrand2016blind}. Accordingly, we use $64$-, $256$-, and $1024$-QAM. For simplicity, we do not consider
any precoding although it may further improve the performance as shown in \cite{ferrand2016blind}. We assume that the phase noises are i.i.d.~with
standard deviation of $1^\circ$. We make the following observations. First, as in the SISO case, phase noise mitigation substantially improves
performance. Also, the proposed likelihood approximation remains accurate as shown by the comparison between the exhaustive search \eqref{eq:MLapp}
and the lower bound on ML detection. Further, the proposed SIW algorithm follows closely the exhaustive search curve and hence achieves near ML
performance. Finally, although in the considered scenario the naive LMMSE is outperformed by the naive ML scheme, the selection between them can
provide a non-negligible gain as shown in Figure~\ref{fig:losmimo0.7}.

\begin{figure}[!t]
  \begin{subfigure}[b]{0.329\textwidth}
\includegraphics[width=\textwidth]{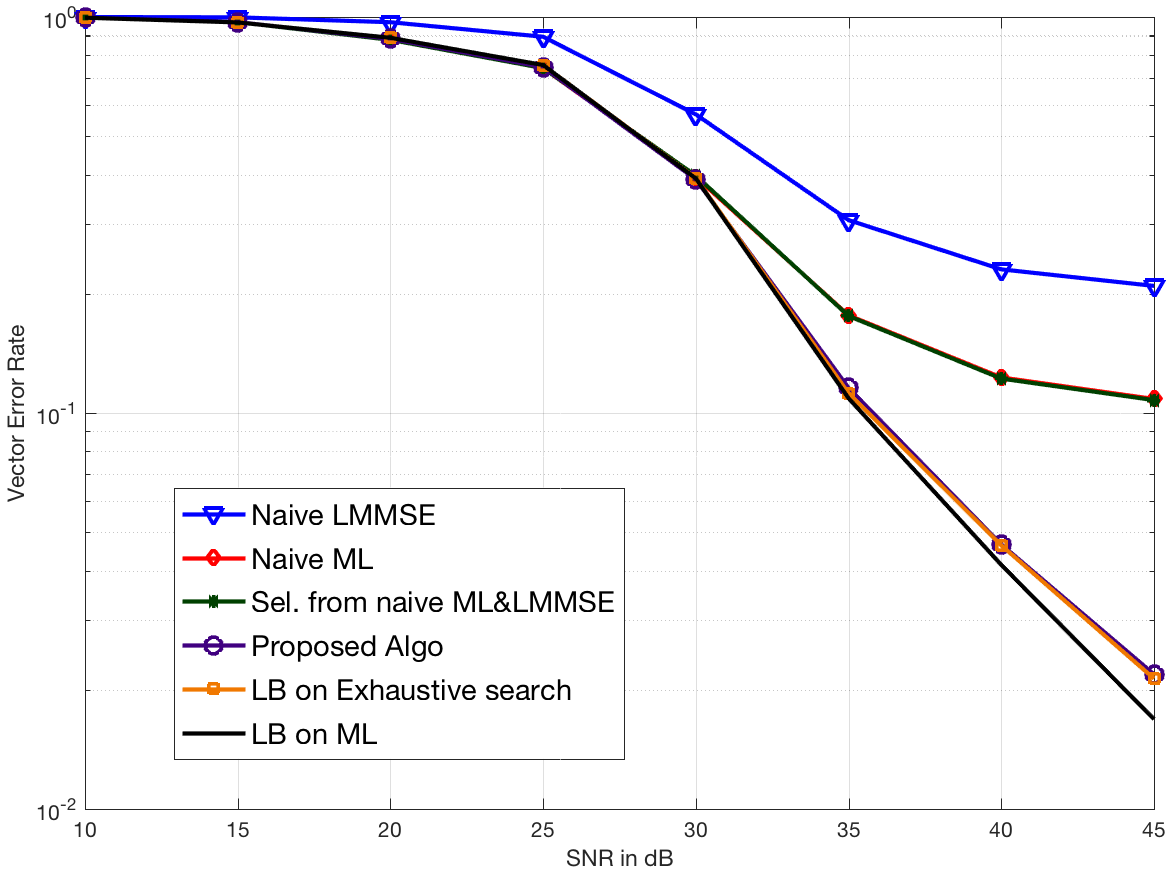}
\caption{0.33 Opt.~distance, $64$-QAM.}
  \end{subfigure}
  \label{fig:losmimo033}
 \begin{subfigure}[b]{0.329\textwidth}
\includegraphics[width=\textwidth]{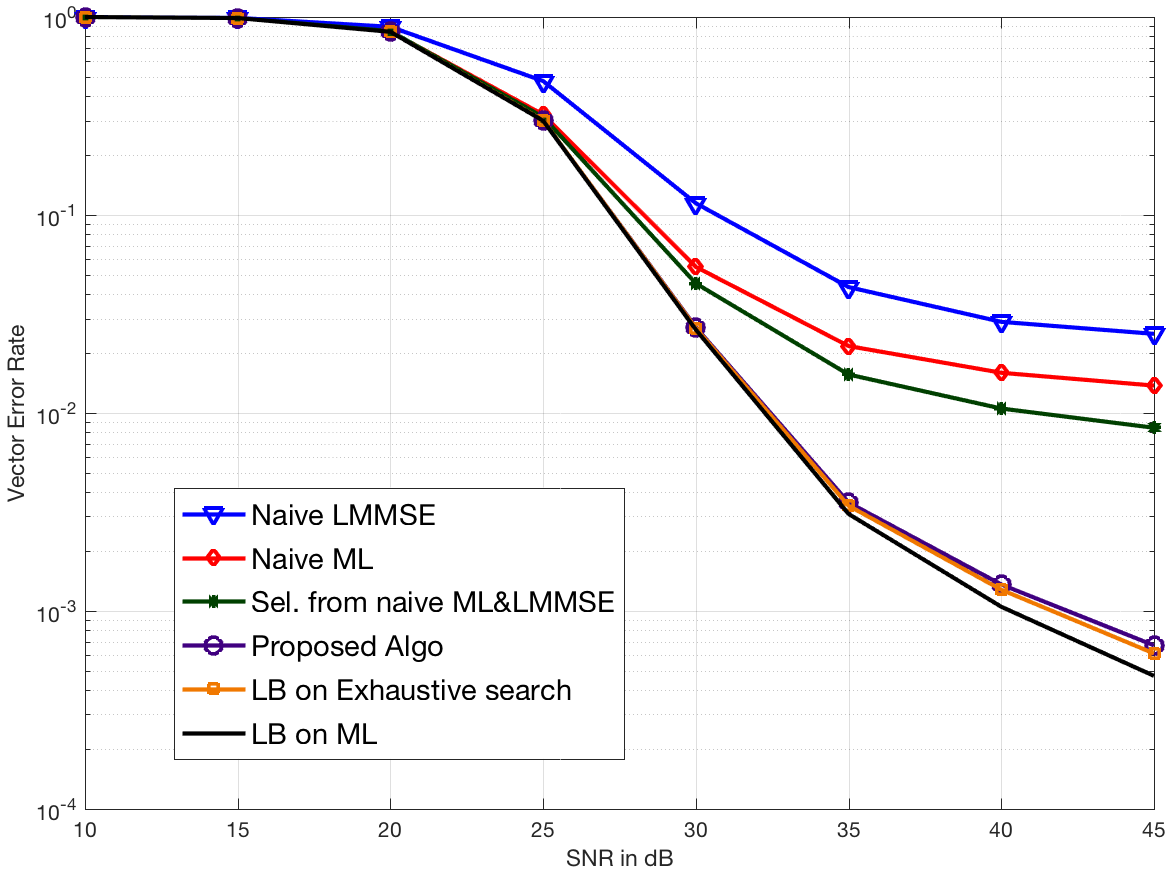}
\caption{0.7 Opt.~distance, $256$-QAM.}
\label{fig:losmimo0.7}
  \end{subfigure}
  \begin{subfigure}[b]{0.329\textwidth}
\includegraphics[width=\textwidth]{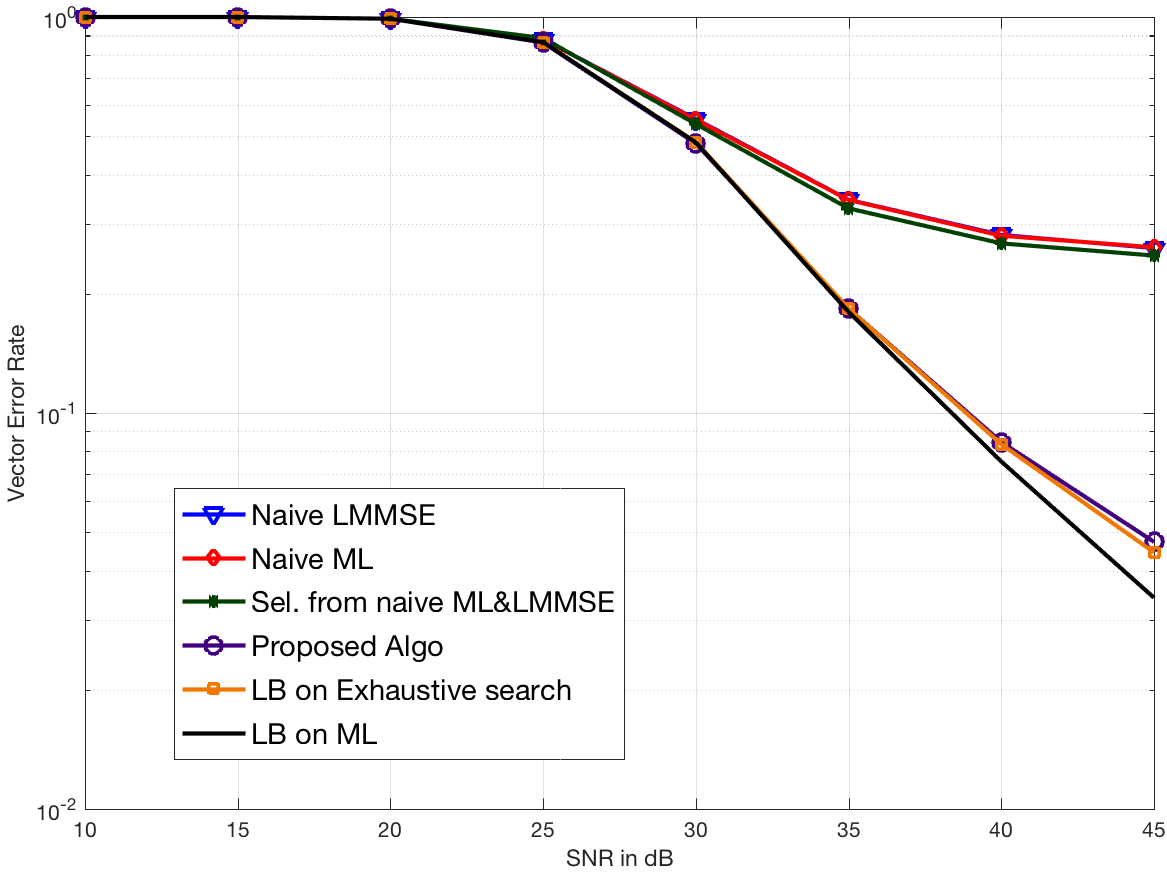}
\caption{Opt.~distance, $1024$-QAM.}
\label{fig:losmimo1}
  \end{subfigure}
 \caption{$4\times4$ LoS MIMO. Each antenna has i.i.d.~phase noise of standard deviation $1^\circ$.}
  \label{fig:LoSMIMO}
\end{figure}

\subsection{Scenario~3: Uplink SIMO channel with centralized receiver oscillator}

\begin{figure}[t]
  \begin{subfigure}[b]{0.329\textwidth}
\includegraphics[width=\textwidth]{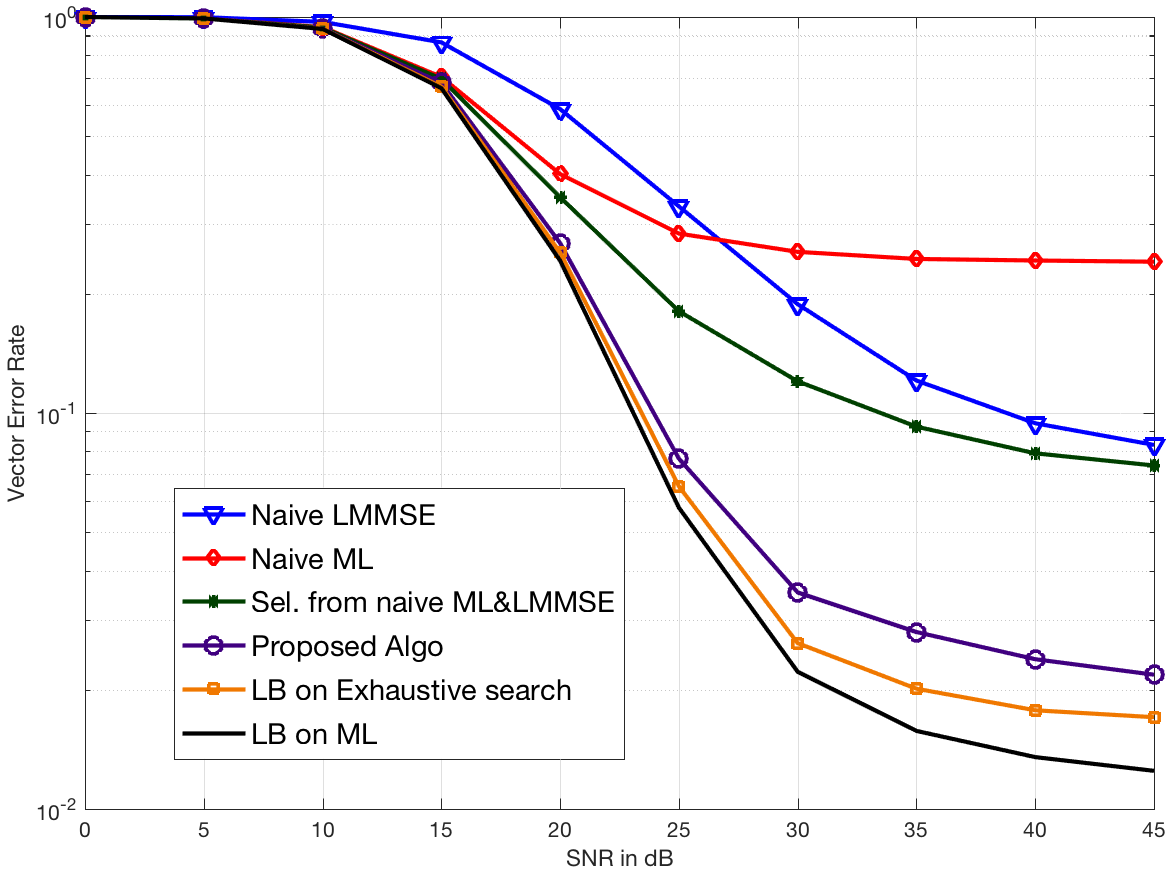}
\caption{$4\times4$, $64$-QAM, PN $4^\circ$.}
\label{fig:simo44-64}
  \end{subfigure}
  \begin{subfigure}[b]{0.329\textwidth}
\includegraphics[width=\textwidth]{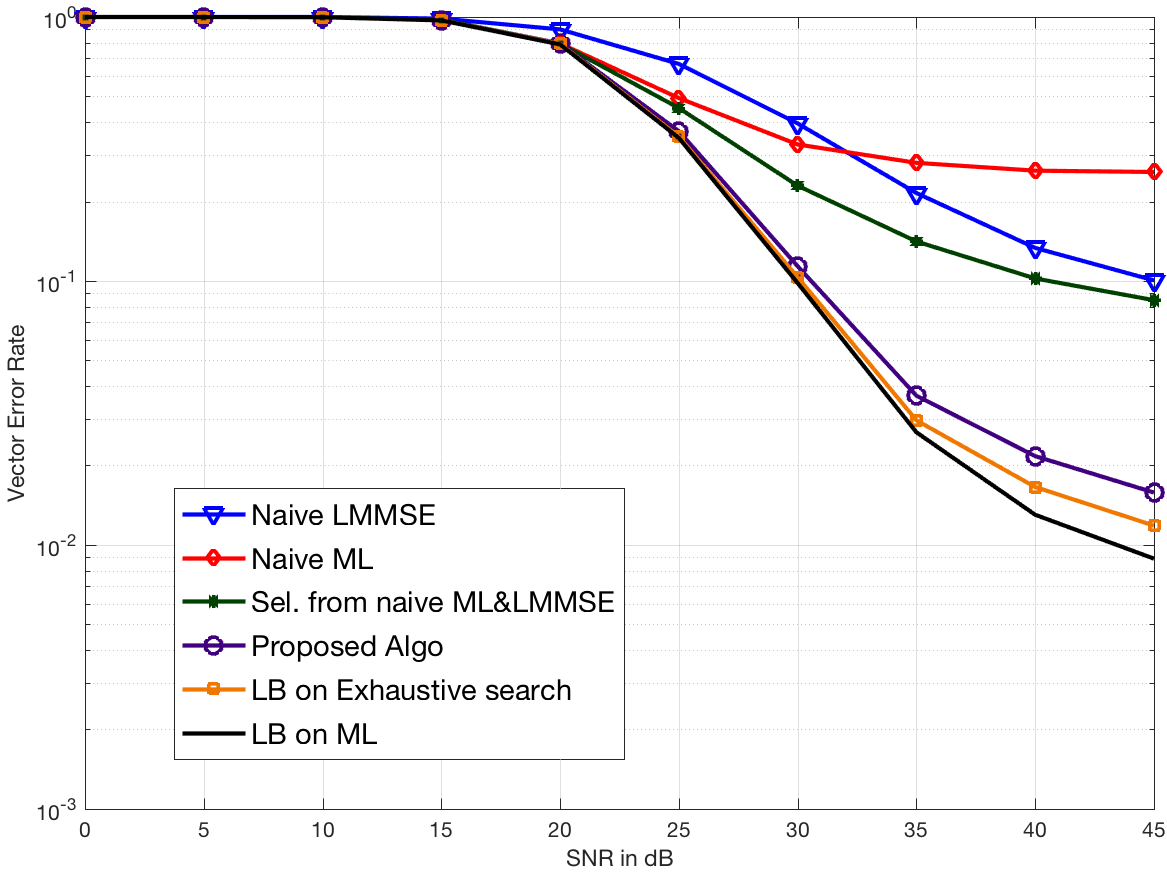}
\caption{$4\times4$, $256$-QAM, PN $2^\circ$.}
\label{fig:simo44-256}
  \end{subfigure}
  \begin{subfigure}[b]{0.329\textwidth}
\includegraphics[width=\textwidth]{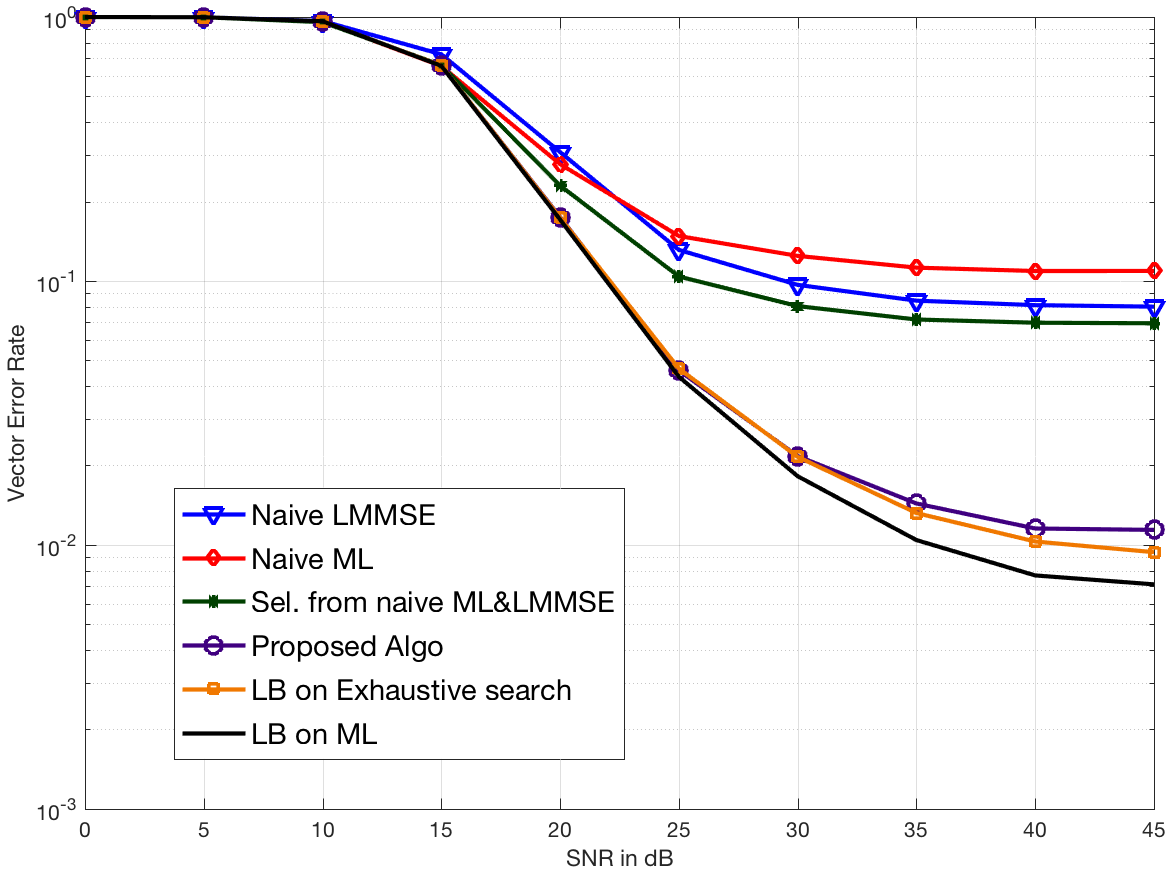}
\caption{$4\times10$, $256$-QAM, PN $2^\circ$.}
  \label{fig:simo410}
  \end{subfigure}
 \caption{Uplink SIMO channel. Four single-antenna users with i.i.d.~phase noise, multi-antenna receiver without phase
 noise.}
  \label{fig:SIMO}
\end{figure}

The third scenario is the uplink cellular communication channel with four single-antenna users and one
multi-antenna base station receiver. It is assumed that the phase noises at the users' side are i.i.d., whereas
there is no phase noise at the receiver side. This is a reasonable assumption since the oscillators at the base
station are usually of higher quality than those used by mobile devices. 
We assume i.i.d.~Rayleigh fading in this scenario where three configurations are considered, as shown in
Figure~\ref{fig:SIMO}. Unlike in the previous scenarios, the naive ML is dominated by the naive LMMSE at high
SNR. This somewhat counter-intuitive observation can be explained as follows.  Without receiver phase noise,
the channel can be inverted and we obtain a spatial parallel channel. Although channel inversion
incurs some power loss when the channel is not orthogonal, each of the resulting parallel subchannels sees an
independent phase noise. Therefore, the demodulation only suffers from a scalar self-interference. On the
other hand, with naive ML the receiver tries to find the closest vector in the image space of $\Hm$ to the
received vector $\yv$. Since the linear map $\Hm$ mixes the perturbation of the different transmit phase
noises, the naive ML detection, ignoring the presence of phase noise, suffers from the aggregated
perturbation from all the phase noises. That is why the naive LMMSE can be better than the naive ML scheme in the
high SNR regime where phase noise dominates the additive noise.  In the case when both the transmitter and the
receiver have comparable phase noises, such a phenomenon is rarely observed since the channel inversion also
increases the perturbation with the presence of receiver phase noises. 

From Figure~\ref{fig:SIMO}, we remark that as before the proposed SIW scheme is superior to the other schemes.
With a relatively strong phase noise of $4^\circ$, the error rate of SIW is $3$ to $4$ times lower
than that of the naive schemes and is at most twice that of the ML lower bound. With a smaller phase noise of $2^\circ$,
the SIW scheme can support $256$-QAM with four receive antennas, achieving an error rate $5$ times lower than that of the naive schemes. In Figure~\ref{fig:simo410}, we increase the number of receive antennas to $10$, we see that
the gap between the naive schemes is decreased due to the increased orthogonality of the channel. Nevertheless,
the gap between the naive schemes and the proposed scheme does not decrease since the orthogonality between the
users does not reduce the impact of the phase noise from the transmitter side. We could expect the same
observation even with massive MIMO. Nevertheless, with massive MIMO uplink, the receiver phase noise can be
mitigated substantially due to the asymptotic orthogonality~\cite{bjornson2014massive}.

\section{Further Discussion and Experiments}
\label{sec:discussions}

\subsection{Robustness to the phase noise distribution} 

\begin{figure}[t]
  \centering
\includegraphics[width=0.5\textwidth]{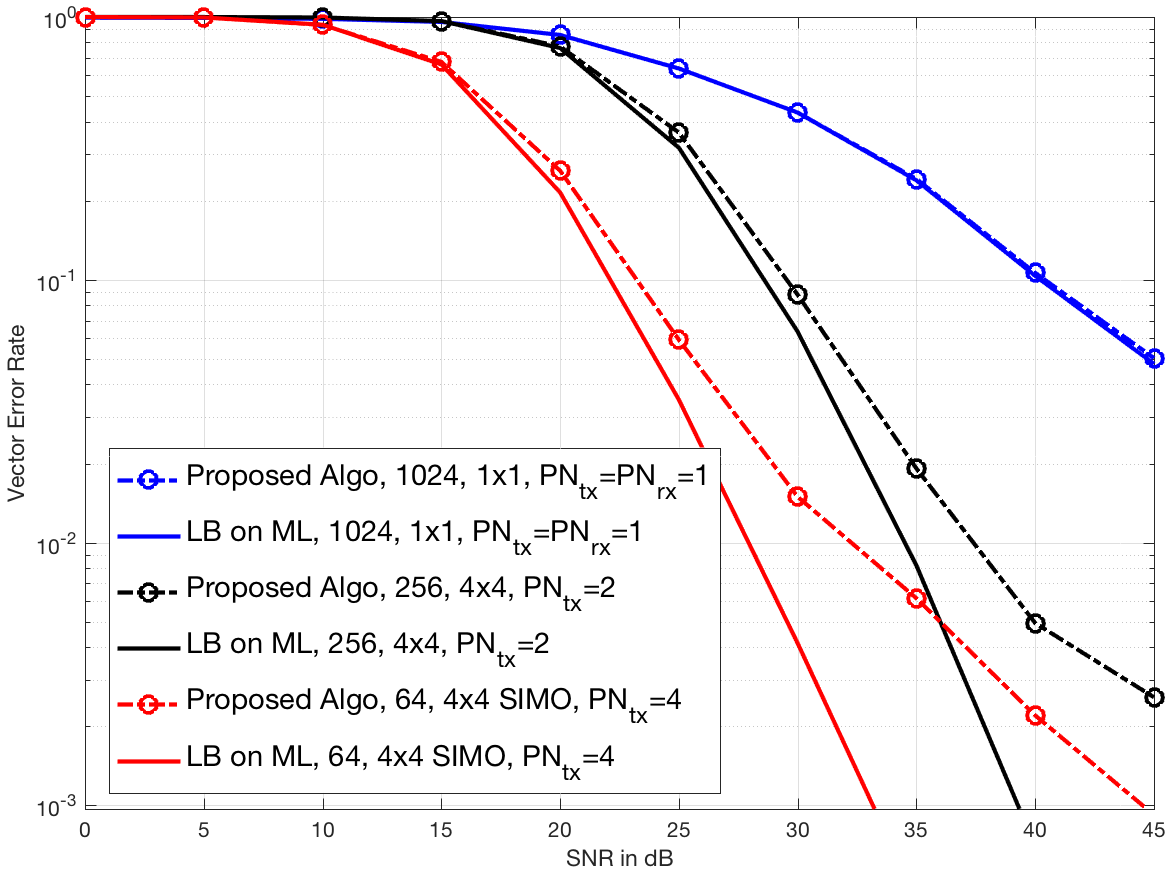}
\caption{Uniform phase noise at both transmitter and receiver, Rayleigh fading.}
\label{fig:uniform}
\end{figure}

One of the main assumptions of our work is that the phase noise follows a Gaussian distribution.  Indeed, our
derivation of the closed-form approximation~\eqref{eq:aML} depends on this assumption. We have seen that this approximation is very accurate in various practical scenarios with Gaussian phase noise. In practice, however, phase noises may not be Gaussian, which leads to the following natural
question on the robustness: \emph{Does the proposed algorithm still work when the phase noise is not Gaussian?}
The answer turns out to be positive when we let the phase noise be uniformly distributed. In
Figure~\ref{fig:uniform}, we consider three previous configurations~(shown in Figure~\ref{fig:SISO1024},
\ref{fig:simo44-64}, and \ref{fig:simo44-256}, respectively) but with uniform phase noises. 
From the results, we see that the proposed algorithm works well as in the Gaussian phase noise, especially when
the phase noise is small. 
In fact, we believe that the phase noise distribution with certain regularity~(e.g., continuous and bounded
density) should not have great impact on the performance of the proposed algorithm when the phase noise
is not too strong.

\subsection{Further improvement with more iterations}

Although the proposed scheme achieves near ML performance in many cases, there is still room for improvement in
some situations. As shown in Figure~\ref{fig:SIMO} for the $4 \times 4$ channel with $64$-QAM, the error floor of
the proposed scheme is three times higher than the one of the exhaustive search. In order to reduce the error
floor, we propose to introduce a list of candidate solutions by iteration. We recall that in the SIW scheme we
start from the best solution between naive ZF and naive ML, and we use this solution as a starting point to estimate the
matrix $\Wm_x$ followed by a nearest neighbor detection. The SIW algorithm replaces the starting point with the
newly found one if the latter has a higher approximated likelihood. We can extend the SIW algorithm with more
iterations. Specifically, we can fix a maximum number of iterations. As long as the newly found point is not
inside the list, it should be added to the list and we continue to iterate. The procedure stops either when we
hit the maximum number of iterations or we find a point already inside the list. At the end, we select the
point with the highest likelihood value from the list.

\begin{figure}[t]
  \centering
\includegraphics[width=0.6\textwidth]{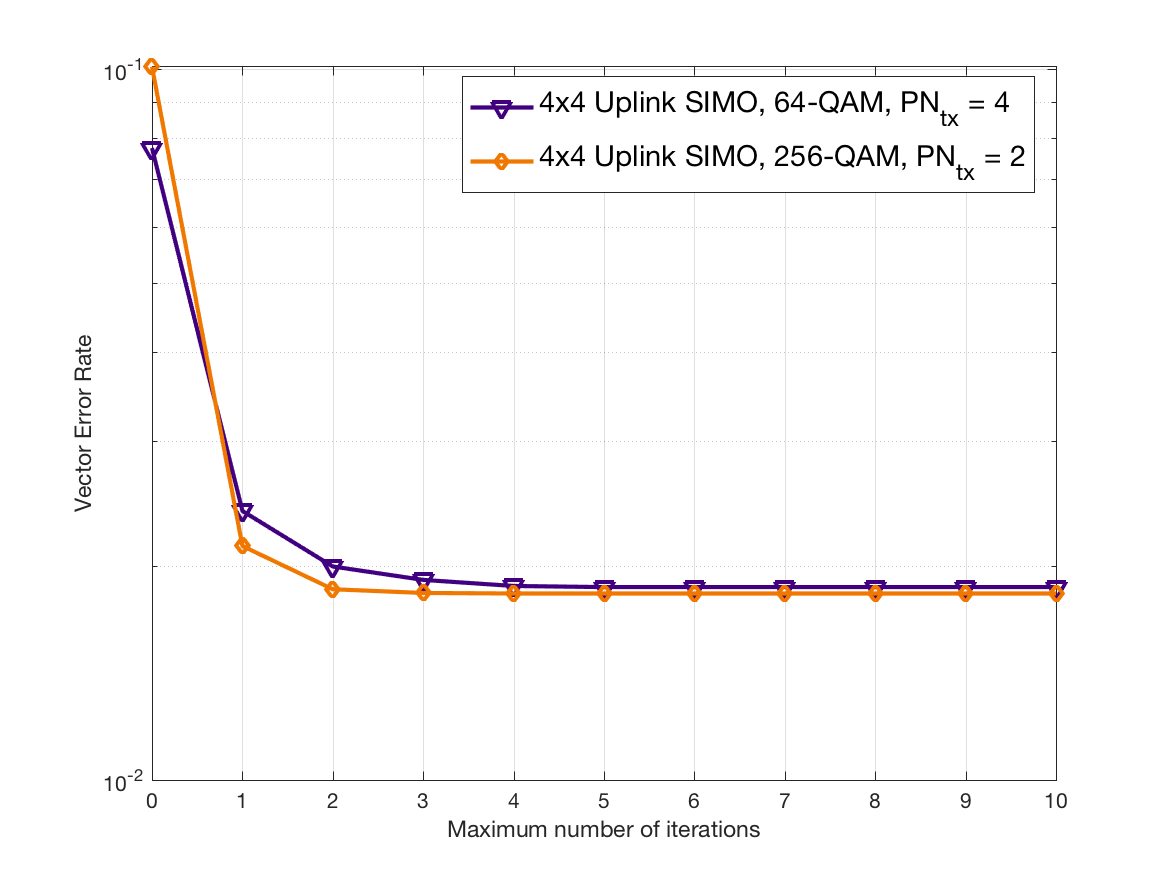}
\caption{Impact of iterations.}
\label{fig:iteration}
\end{figure}

\subsection{Potential issues in the very high SNR regime}

At very high SNR, the proposed approximation~\eqref{eq:aML} may become less accurate and the performance loss
increases with SNR. It can be seen from \eqref{eq:llh0} that the approximation error of
$\|\yv-\Hm_{\!\rvVec{\Theta}}\,\xv\|$ has greater impact as $\SNR$ grows. Numerical experiments show that such
performance loss may be apparent for SNR higher than $40$dB depending on the constellation size. A
straightforward workaround is to impose a ceiling value of the SNR $\SNR_{\max}$ for the decoder. In other words,
for SNR higher than the threshold $\SNR_{\max}$, we let the decoder work as if it were at $\SNR_{\max}$.
Intuitively, the probability of error cannot be larger than the one at $\SNR_{\max}$, since we are feeding less
noisy observations to the decoder than what they are supposed to be. With the
decoding function, the observation space can be partitioned into $|\mathcal{X}|^{\nt}$ regions, each one of
which corresponds to a vector in $\mathcal{X}^{\nt}$. The probability of decoding error is the probability that
the observation is outside of the decoding region corresponding to the actual input. With a smaller variance of
AWGN, the observation has a higher probability to be inside the region, hence a lower probability of error. 
Nevertheless, the formal proof of this argument is not trivial, and is outside of the scope of the current
paper.\footnote{We can always add an artificial noise to the observation in order to reduce the SNR to
$\SNR_{\max}$ before the detection. In this way we have an error floor for SNR beyond $\SNR_{\max}$.}  

Another issue at very high SNR is that the likelihood may be too small as compared to the finite
numerical precision. Therefore, it becomes impossible to obtain the simulation-based lower bound in a reliable
way. Furthermore, the number of Monte-Carlo samples required to reach any given accuracy grows with the SNR.

\subsection{On the practical validity of the adopted discrete-time model}
\label{sec:validity}

The discrete-time channel model~\eqref{eq:input-output} that we adopt in this work is a
simplification of the waveform phase noise channel. Indeed, the discrete-time output sequence is
obtained from filtering followed by sampling in most communication systems. Filtering a waveform
corrupted by phase noise results in not only phase perturbation but also amplitude
variation~\cite{Foschini-PN}. In the following, we shall show that the amplitude variation is
negligible in the practical regime of interest. For simplicity, we focus on the single-antenna
case with a rectangular filter~(i.e., an integrator). We adhere to the commonly accepted Wiener
model for the continuous-time phase noise process $\{\rV{\Theta}(t)\}$. In particular,
$\rV{\Theta}(t)\sim \mathcal{N}(0,\beta t)$. The equivalent filtered channel gain for the $k$~th
symbol interval of duration $T$ is
\begin{align}
\frac{1}{T} \int_{kT}^{(k+1)T} e^{j \rV{\Theta}(t)} \mathrm{d} t 
&\overset{d}{=} e^{j \rV{\Theta}(kT)} \frac{1}{T} \int_{0}^{T} e^{j \rV{\Theta}(t)} \mathrm{d} t
\end{align}%
from the property of a Wiener process where $\overset{d}{=}$ means equality in distribution.
Assuming that we can somehow track the \emph{past state} $\rV{\Theta}(kT)$ perfectly, we now focus
on the following random variable due to the residual phase noise corresponding to the
\emph{innovation part}: 
\begin{align}
  \rV{B}(\beta,T) :\!\!&= \frac{1}{T} \int_{0}^{T} e^{j \rV{\Theta}(t)} \mathrm{d} t \label{eq:def-B}\\
  &\overset{d}{=} \frac{1}{\beta T} \int_{0}^{\beta T} e^{j \tilde{\rV{\Theta}}(t)} \mathrm{d} t
  \label{eq:tmp88}\\
  &\overset{d}{=} \rV{B}(1, \beta T)
\end{align}%
where in \eqref{eq:tmp88} follows from $\rV{\Theta}(t) \overset{d}{=} \tilde{\rV{\Theta}}(\beta t)
$ for some normalized Wiener process $\{\tilde{\rV{\Theta}}(t)\}$ with $\tilde{\rV{\Theta}}(t)\sim
\mathcal{N}(0,t)$. We notice that the random variable $\rV{B}(\beta,T)$ depends on the parameters
$\beta$ and $T$ only through the product $\beta T$. The distribution of $\rV{B}(\beta,T)$
has been characterized both approximately~\cite{Foschini-PN} and exactly~\cite{wang2006solving}.
In particular, an approximate moment-generating function has been derived in \cite{Foschini-PN}
for small phase noise. We can use the results therein to obtain the following characterization.  
\begin{proposition} \label{prop:MvsPhi}
  Let us consider the polar representation $\rV{B}(\beta,T) = \rV{G} e^{j \rV{\Phi}}$ with
  $\rV{G}\ge0$ and $\rV{\Phi}\in[-\pi,\pi)$. 
  When $S := \beta T$ is small, we have
  \begin{align}
    \Var\left( \rV{\Phi} \right) &\approx \frac{S}{3}, \quad
    \Var\left( \rV{G} \right) \approx \frac{S^2}{180}, \label{eq:Var}\\
    \intertext{and thus,}
    \Var\left( \rV{G} \right) &\approx \frac{1}{20} \bigl[ \Var\left(
    \rV{\Phi} \right) \bigr]^2. \label{eq:MvsPhi}
  \end{align}%
\end{proposition}
In Figure~\ref{fig:wiener}, we compare the approximation given by~\eqref{eq:MvsPhi} to the correct value obtained numerically. The approximation~\eqref{eq:MvsPhi} is surprisingly accurate even for a standard deviation of $20^\circ$ for the phase. More importantly, such results show that for a
small phase perturbation up to $5$ degrees -- this is the regime of interest in the present work --
the amplitude variation is less than $-55$~dB. Therefore, the interference caused by the
amplitude variation is dominated by the AWGN and can be treated as noise without any performance
loss. In other words, the discrete-time model adopted in this work -- ignoring the amplitude
variation -- is indeed valid for phase noise with moderate variance.  

\begin{figure}[t]
  \centering
\includegraphics[width=0.7\textwidth]{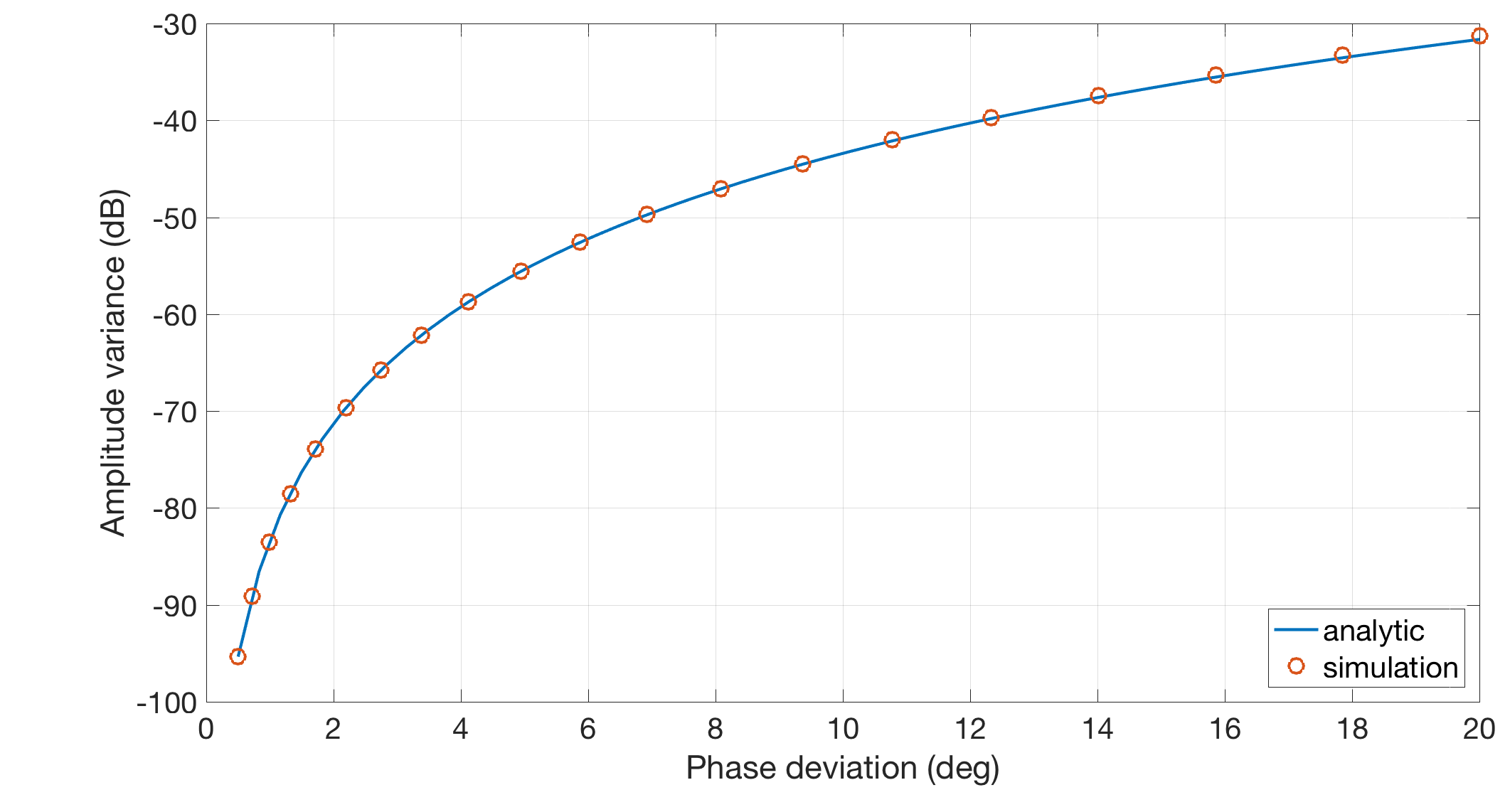}
\caption{Amplitude vs phase perturbation for filtered channel gain.}
\label{fig:wiener}
\end{figure}

\subsection{Constellation design}
The purpose of the current work is to design a detector that takes into account the phase
perturbation for existing systems in which typical QAM signaling is used. As we observe from the
numerical experiments, in some of the scenarios, even the lower bound on the ML detection error
exhibits an error floor. The error floor is however not a fundamental
limit of either the detector or the channel. With a carefully designed signaling scheme, the
probability of ML detection error can be arbitrarily small with an increasing SNR. For instance,
if we use amplitude modulation, then when the SNR grows, the amplitude ambiguity
decreases and the detection error vanishes. The cost of such an extreme scheme is a reduced
spectral efficiency. Algorithms to design a constellation based on the statistics of both the
additive and phase noises have been proposed in the literature~(see
\cite{Foschini-PN-constellation, Amat-constellation} and the references therein). To the best of the authors' knowledge, only SISO has been
considered so far. In the MIMO case, such constellations can still be used and should provide
improvements over QAM. The difficulty with non-QAM constellation lies in the MIMO detection part,
since efficient NND algorithms cannot be applied directly. The constellation design problem for
MIMO phase noise channels is a challenging and interesting problem in its own right, which is however
out of the scope of the current work.

\section{Conclusions}
\label{sec:conclusions}

In this work, we have studied the ML detection problem for uncoded MIMO phase noise channels. We have proposed
an approximation of the likelihood function that has been shown to be accurate in the regimes of practical
interest. More importantly, thanks to the geometric interpretation of the approximate likelihood function, we
have designed a simple algorithm that can solve approximately the optimization problem with only two nearest
neighbor detections. Numerical experiments show that the proposed algorithm can greatly mitigate the impact of phase noises
in different communication scenarios.  

\bibliographystyle{IEEEtran}
\bibliography{./biblio}

\allowdisplaybreaks[0]

\appendix
\subsection{Proof of Proposition~\ref{th:radius}}

Define the vector $\rvVec{V} := (\rvLambdaR \rvMat{H} \rvLambdaT - \rvMat{H}) \rvVec{X}$ so that
$\rV{R}^2 = \| \rvVec{V} + \rvVec{Z} \|^2$. We have:
	\begin{align}
          \rV{R}^2 &=  \| \rvVec{V} \|^2 +  \| \rvVec{Z} \|^2 + \rvVec{Z}^H \rvVec{V} + \rvVec{V}^H \rvVec{Z} \\
          \rV{R}^4 &= ( \| \rvVec{V} \|^2 +  \| \rvVec{Z} \|^2)^2 + (\rvVec{Z}^H \rvVec{V} + \rvVec{V}^H \rvVec{Z})^2 + 2 ( \| \rvVec{V} \|^2 +  \| \rvVec{Z} \|^2)^2(\rvVec{Z}^H \rvVec{V} + \rvVec{V}^H \rvVec{Z}).
	\end{align}
        Using the fact that $\rvVec{V}$ and $\rvVec{Z}$ are independent and $\rvVec{Z}$ has
        i.i.d.~$\mathcal{CN}(0,\gamma^{-1})$ entries, we verify that
	\begin{align}
          \E[\rV{R}^2] &=  \E[\| \rvVec{V} \|^2] +  \E[\| \rvVec{Z} \|^2], \label{eq:ER2}\\
          \E[\rV{R}^4] &= \E[\| \rvVec{V} \|^4] + \E[\| \rvVec{Z} \|^4]  + 2 \E[\| \rvVec{Z}
          \|^2]\E[\| \rvVec{V} \|^2] + 2 \gamma^{-1} \|\rvVec{V}\|^2, \\
          \Var(\rV{R}^2) &= \Var(\| \rvVec{V} \|^2) + \Var(\| \rvVec{Z} \|^2) + 2 \gamma^{-1} \E
          \|\rvVec{V}\|^2. \label{eq:VarR2}
	\end{align}
        Since $\E[\| \rvVec{Z} \|^2] = \nr \gamma^{-1}$ and $\Var(\| \rvVec{Z} \|^2) = 2 \nr
        \gamma^{-2}$,
        it follows that 
	\begin{align}
          \E[\rV{R}^2] &=  \E[\| \rvVec{V} \|^2] + \gamma^{-1}\, \nr, \\
          \Var(\rV{R}^4) &= \Var(\| \rvVec{V} \|^2) + 2 \gamma^{-2} \nr + 2 \gamma^{-1} \E[\|\rvVec{V}\|^2].
	\end{align}
	 We now calculate the moments of $\|\rvVec{V}\|^2$.
	 
         {\bf Expectation over $\rvMat{H}$}: Let us define $\rV{A}_{k,l} := \rV{X}_l(
         e^{j(\rV{\Theta}_{t,l} + \rV{\Theta}_{r,k})} - 1)$ and write, conditional on $\rvMat{A}$, 
	\begin{equation}
          \| \rvVec{V} \|^2 =  \sum_{k=1}^{\nr} \left| \sum_{l=1}^{\nt} \rV{H}_{k,l} \rV{A}_{k,l}
          \right|^2 \overset{d}{=} \sum_{k=1}^{\nr} \rV{E}_k,
	\end{equation}%
        where $\rV{E}_k\sim \mathrm{Exp}(1/\|\rvVec{A}_{k}\|^2)$ with $\rvVec{A}_k :=
        [\rV{A}_{k,l}]_{l=1,\ldots,\nt}$, $k=1,\ldots,\nr$;
        we used the fact that $\rvMat{H}$ has i.i.d.~$\mathcal{CN}(0,1)$ entries. 
        Then we can calculate the first and second moments of $\| \rvVec{V} \|^2$ for a given 
        $\rvMat{A} = [\rV{A}_{j,k}]_{j,k}$: 
        \begin{align}
          \E_{\rvMat{H}|\rvA}[\| \rvVec{V} \|^2] &= \sum_{k=1}^{\nr} \E_{\rvMat{H}|\rvA}[
          \rV{E}_k ] = \sum_{k=1}^{\nr} \sum_{l=1}^{\nt} |\rV{A}_{k,l}|^2, \label{eq:EVA} \\
          \E_{\rvMat{H}|\rvA}\left[ \| \rvVec{V} \|^4 \right] &= \E_{\rvMat{H}|\rvA} \left[ \sum_{k=1}^{\nr} \rV{E}_k^2 \right] + \E_{\rvMat{H}|\rvA} \left[ \sum_{k'=1\atop k'\ne k}^{\nt} \sum_{k=1}^{\nt} \rV{E}_k \rV{E}_{k'}  \right] \\
          &= \sum_{k=1}^{\nr} 2 \|\rvVec{A}_{k}\|^4  + \left[ \sum_{k'=1\atop k'\ne k}^{\nt} \sum_{k=1}^{\nt}  \|\rvVec{A}_{k}\|^2 \|\rvVec{A}_{k'}\|^2 \right] \\
          &=  \sum_{k'=1}^{\nr} \sum_{k=1}^{\nr}  (1 + \indic\{ k=k'\}) \|\rvVec{A}_{k}\|^2
          \|\rvVec{A}_{k'}\|^2.  \label{eq:EV2A}
        \end{align}%

        {\bf Moments of $\rvA$:} We recall that $|\rV{A}_{k,l}|^2 = 2 |\rV{X}_l|^2
        (1-\cos(\rV{\Theta}_{r,k} + \rV{\Theta}_{t,l}))$. We have that $\E\bigl[|\rV{X}_l|^2\bigr]
        = 1$
        and define $\E\bigl[|\rV{X}_l|^4\bigr] = \bar{P}^2$, for $l=1,\ldots,\nt$. Using the
        independence and the identity $\E[\cos(\rV{\Theta})] =
        \exp\bigl(-\frac{\Var(\rV{\Theta})}{2}\bigr)$ for zero-mean Gaussian $\rV{\Theta}$,  we
        obtain
	\begin{equation}
          \E[|\rV{A}_{k,l}|^2] = 2 \E[|\rV{X}_l|^2] \E[1 - \cos(\Theta_{t,l} + \Theta_{r,k})] = 2 
          \Bigl(1 - e^{-{\sigma_r^2 + \sigma_t^2 \over 2}}\Bigr). \label{eq:EA}
	\end{equation}
	We now calculate the correlation between the entries of $\rvA$
	\begin{equation}	
          \E(|\rV{A}_{k,l}|^2|\rV{A}_{k',l'}|^2) = 4
          \E(|\rV{X}_{l}|^2|\rV{X}_{l'}|^2)\,\rho_{k,k'\!\!,\,l,l'} 
	\end{equation}
		with
	\begin{align}
		\rho_{k,k'\!\!,\,l,l'} :\!\!&= \E(  (1 - \cos(\Theta_{t,l} + \Theta_{r,k}))(1 - \cos(\Theta_{t,l'} + \Theta_{r,k'}))) \\
                &= 1 - 2 e^{-{\sigma_r^2 + \sigma_t^2 \over 2}} + e^{-\sigma_r^2 - \sigma_t^2}
                \cosh( \sigma_r^2(\indic\{k = k'\} + \sigma_t^2 \indic\{l = l'\}) ),
	\end{align}
        where to obtain the last equality we use the trigonometric identities and again apply the
        identity $\E[\cos(\rV{\Theta})] = \exp\bigl(-\frac{\Var(\rV{\Theta})}{2}\bigr)$; we recall
        that $\cosh(x) = \frac{e^{x}+e^{-x}}{2}$.

	{\bf Moments of $\| \rvVec{V} \|^2$:} 
        From \eqref{eq:EVA} and \eqref{eq:EA}, we have the first moment
	\begin{equation}
          \E[\| \rvVec{V} \|^2] = \sum_{k=1}^{\nr} \sum_{l=1}^{\nt} \E[|\rV{A}_{k,l}|^2] = \nt \nr
          2 \Bigl(1 - e^{-{\sigma_r^2 + \sigma_t^2 \over 2}}\Bigr). \label{eq:EV}
	\end{equation}
        For the variance, we apply \eqref{eq:EV2A} and \eqref{eq:EV} 
	\begin{align}
          \Var\left( \| \rvVec{V} \|^2 \right) &= \E[\| \rvVec{V} \|^4] - \bigl( \E[\| \rvVec{V}
          \|^2]\bigr)^2 \\ 
          &= \sum_{k'=1}^{\nr} \sum_{k=1}^{\nr} \sum_{l'=1}^{\nt} \sum_{l=1}^{\nt} \Bigl( \E[|\rV{A}_{k,l}|^2
          |\rV{A}_{k',l'}|^2] (1 + \indic\{ k=k'\}) - \E[|\rV{A}_{k,l}|^2] \E[|\rV{A}_{k',l'}|^2]
          \Bigr). \label{eq:sum}
	\end{align}
        Noting that $\E[|\rV{A}_{k,l}|^2 |\rV{A}_{k',l'}|^2] = \E[|\rV{A}_{k,l}|^2]
        \E[|\rV{A}_{k',l'}|^2]$ if $k\ne k'$ and $l \ne l'$, we obtain the variance
	\begin{equation}
		\Var(\| \rvVec{V} \|^2) = 4\nt\nr\Bigl( w_1 +
                w_2  (\nt-1) + w_3 (\nr-1) \Bigr), \label{eq:VarV2}
	\end{equation}
        where $w_1,w_2,w_3>0$ do not depend on $\nt,\nr$ and correspond to the cases $(k=k',l=
        l')$, $(k=k',l\ne l')$, $(k\ne k',l= l')$, respectively, in the summation \eqref{eq:sum},
	\begin{align}
		w_1 &:= 2 \bar{P}^2 (1 - 2 e^{-\sigma^2} + e^{-2 \sigma^2} \cosh( 2 \sigma^2)) -
                (1 - e^{-\sigma^2})^2, \\ 
	  w_2 &:= 2  ( 1 - 2 e^{-\sigma^2} + e^{-2 \sigma^2} \cosh(\sigma_r^2)) - (1 -
          e^{-\sigma^2})^2, \\
	  w_3 &:= \bar{P}^2 ( 1 - 2 e^{-\sigma^2} + e^{-2 \sigma^2} \cosh(\sigma_t^2)) - (1 -
          e^{-\sigma^2})^2,
	\end{align}
        where we define $\sigma^2 := {\sigma_t^2 + \sigma_r^2 \over 2}$. 
	
        {\bf Putting it together:} From \eqref{eq:ER2} and \eqref{eq:EV}, we have
	\begin{equation}
          \E[\rV{R}^2] = \E[\| \rvVec{V} \|^2]  + \gamma^{-1} \nr = \nt \nr 2 \Bigl(1 -
          e^{-{\sigma_r^2 + \sigma_t^2 \over 2}}\Bigr) + \gamma^{-1} \nr,
	\end{equation}
        which yields the first result. Note that $\E[\rV{R}^2] \ge c_1 \nt\nr$ for some constant
        $c_1>0$ with respect to $(\nt,\nr)$. From \eqref{eq:VarR2} and \eqref{eq:VarV2}, we have
        proven that there exists a constant $c_2>0$ such that
	\begin{equation}
          \Var(\rV{R}^2) = \Var(\| \rvVec{V} \|^2) + 2 \gamma^{-2} \nr + 2 \gamma^{-1} \E[\|
          \rvVec{V} \|^2] \le  c_2 \nr \nt (\nt + \nr).
	\end{equation}
	Hence
	\begin{equation}
          {\Var( \rV{R}^2) \over (\E[ \rV{R}^2])^2} \le
          \frac{c_2}{c_1^2}\left( {1 \over \nr} + {1 \over \nt} \right)  \to
          0, \quad \nt,\nr \to \infty,
	\end{equation}
	and applying Chebychev's inequality yields the second result:
	\begin{equation}
          \PP\left\{ (1-\eta) \E[\rV{R}^2] \le \rV{R}^2 \le (1+\eta) \E[\rV{R}^2]\right\} \to 1, \nt,\nr \to \infty.
	\end{equation}
	
\subsection{Proof of Proposition~\ref{prop:MvsPhi}}
\newcommand{\sinhc}{\mathrm{sinhc}}
\newcommand{\tanhc}{\mathrm{tanhc}}
\newcommand{\cothc}{\mathrm{cothc}}
From \cite[eq.6]{Foschini-PN}\footnote{Note that the random variable $B$ in \eqref{eq:def-B}
differs with the one in \cite[eq.1]{Foschini-PN} in a normalization factor $T$. The MGF has been
scaled accordingly.}, we can derive the moment-generating function~(MGF) of
$(\rV{G},\rV{\Phi})$
\begin{multline}
  M_{\rV{G},\rV{\Phi}}(\xi, \eta) = e^{\xi}\,\sinhc^{-\frac{1}{2}}\left( \sqrt{S \xi}
  \right) \\ \cdot \exp\left[ \frac{S \eta^2}{8} \left[ \cothc\left( \sqrt{\frac{S\xi}{4}} \right)
  - \left( \frac{S\xi}{4} \right)^{-1} + \tanhc\left( \sqrt{S \xi} \right) \right] \right],
  \label{eq:mgf}
\end{multline}%
where we define $\sinhc(x):=\frac{\sinh(x)}{x}$, $\cothc(x):=\frac{\coth(x)}{x}$,
$\tanhc(x):=\frac{\tanh(x)}{x}$, and recall that $S:=\beta T$. It follows that the MGF of
$\rV{G}$ and $\rV{\Phi}$ are 
\begin{align}
  M_{\rV{G}}(\xi) &= M_{\rV{G},\rV{\Phi}}(\xi, 0) = e^{\xi}\sinhc^{-\frac{1}{2}}\left( \sqrt{S
  \xi} \right), \label{eq:mgf-G} \\
  M_{\rV{\Phi}}(\eta) &= M_{\rV{G},\rV{\Phi}}(0, \eta) = \exp\left( \frac{1}{6}S\eta^2 \right),
  \label{eq:mgf-Phi}
\end{align}%
where we used the fact that, when $x\to0$, $\sinhc(x) = 1 + O(x^2)$, $\tanhc(x) = 1 + O(x^2)$, and
\begin{align}
  \cothc(x) - \frac{1}{x^{2}}  &= \frac{1}{3} + O(x^2). 
\end{align}%
After finding the first and second derivatives of both MGF \eqref{eq:mgf-G} and \eqref{eq:mgf-Phi}
with some elementary manipulations, we obtain the desired variances
\begin{align}
  \Var\left( \rV{G} \right) &= \E\left( \rV{G}^2 \right) - \left(\E\left( \rV{G}\right) \right)^2
  = M_{\rV{G}}''(0) - \left( M_{\rV{G}}'(0) \right)^2 = \frac{S^2}{180}, \\
  \Var\left( \rV{\Phi} \right) &= \E\left( \rV{\Phi}^2 \right) - \left(\E\left( \rV{\Phi}\right)
  \right)^2 = M_{\rV{\Phi}}''(0) - \left( M_{\rV{\Phi}}'(0) \right)^2 = \frac{S}{3}. 
\end{align}%
Note that we use approximate equality in \eqref{eq:Var} and \eqref{eq:MvsPhi} since the MGF
derived in \cite{Foschini-PN} is indeed approximative with the assumption of small $S$.

\end{document}